\documentclass[12pt, draftclsnofoot, letterpaper, onecolumn]{IEEEtran}
\usepackage{amsfonts}
\usepackage{amssymb}
\usepackage{graphicx}
\usepackage{subfigure}
\usepackage{enumerate}
\usepackage{amsmath}
\usepackage{color}
\usepackage{amsthm}
\usepackage{amsmath}
\usepackage{algorithm}
\usepackage{algpseudocode}

\ifodd 1

\else

\fi

\ifodd 1
\newcommand{\congc}[1]{{\color{red}(Cong: #1)}}
\else
\newcommand{\congc}[1]{}
\fi

\newcommand{\bp}{\begin{proof} \small }
\newcommand{\ep}{\end{proof} \normalsize}
\newcommand{\epx}{\end{proof} \small}
\newcommand{\bpa}{\begin{proofappx} \footnotesize }
\newcommand{\epa}{\end{proofappx} \small }
\newtheorem{theorem}{Theorem}

\newtheorem{lemma}{Lemma}

\newtheorem{definition}{Definition}

\newtheorem*{theorem*}{Theorem}
\newtheorem*{proposition*}{Proposition}
\newtheorem*{corollary*}{Corollary}
\newtheorem*{lemma*}{Lemma}
\newtheorem*{assumption*}{Assumption}
\newtheorem*{definition*}{Definition}
\newtheorem*{claim*}{Claim}

\newcommand{\bm}[1]{\mbox{\boldmath $#1$}}

\newcommand{\be}{\begin{equation}}
\newcommand{\ee}{\end{equation}}
\newcommand{\bs}{\begin{subequations}}
\newcommand{\es}{\end{subequations}}
\newcommand{\bq}{\begin{eqnarray}}
\newcommand{\eq}{\end{eqnarray}}
\newcommand{\bqn}{\begin{eqnarray*}}
\newcommand{\eqn}{\end{eqnarray*}}

\newcommand{\ba}{\left[ \begin{array}}
\newcommand{\ea}{\\ \end{array} \right]}
\newcommand{\ben}{\begin{enumerate}}
\newcommand{\een}{\end{enumerate}}

\def\B{{\boldsymbol{B}}}

\def\e{{\boldsymbol{e}}}

\def\g{{\boldsymbol{g}}}

\def\p{{\boldsymbol{p}}}

\def\real{{\mathchoice%
{\hbox{\rm\setbox1=\hbox{I}\copy1\kern-.45\wd1 R}}
{\hbox{\rm\setbox1=\hbox{I}\copy1\kern-.45\wd1 R}}
{\hbox{\scriptsize\rm\setbox1=\hbox{I}\copy1\kern-.45\wd1 R}}
{\hbox{\scriptsize\rm\setbox1=\hbox{I}\copy1\kern-.45\wd1 R}}}}

\def\Zint{{\mathchoice{\setbox1=\hbox{\sf Z}\copy1\kern-.75\wd1\box1}
{\setbox1=\hbox{\sf Z}\copy1\kern-.75\wd1\box1}
{\setbox1=\hbox{\scriptsize\sf Z}\copy1\kern-.75\wd1\box1}
{\setbox1=\hbox{\scriptsize\sf Z}\copy1\kern-.75\wd1\box1}}}
\newcommand{\complex}{ \hbox{\rm C\kern-0.45em\rule[.07em]{.02em}{.58em}%
\kern 0.43em}}

\begin{document}
%
\title{Online Geographical Load Balancing for Mobile Edge Computing with Energy Harvesting}

\author{Jie~Xu, Hang~Wu, Lixing~Chen, Cong~Shen
\thanks{J. Xu and L. Chen are with the Department of Electrical and
	Computer Engineering, University of Miami, USA. Email: jiexu@miami.edu, lx.chen@miami.edu.}
\thanks{H. Wu and C. Shen are with the Department of Electric Engineering and Information Science, University of Science and Technology of China. Email: wh1202@mail.ustc.edu.cn, congshen@ustc.edu.cn}
}


\maketitle

\begin{abstract}
Mobile Edge Computing (MEC) (a.k.a. fog computing) has recently emerged to enable low-latency and location-aware data processing at the edge of mobile networks. Since providing grid power supply in support of MEC can be costly and even infeasible in some scenarios, on-site renewable energy is mandated as a major or even sole power supply. Nonetheless, the high intermittency and unpredictability of energy harvesting creates many new challenges of performing effective MEC. In this paper, we develop an algorithm called GLOBE that performs joint geographical load balancing (GLB) and admission control for optimizing the system performance of a network of MEC-enabled and energy harvesting-powered base stations. By leveraging and extending the Lyapunov optimization with perturbation technique, GLOBE operates online without requiring future system information and addresses significant challenges caused by battery state dynamics and energy causality constraints. Moreover, GLOBE works in a distributed manner, which makes our algorithm scalable to large networks. We prove that GLOBE achieves a close-to-optimal system performance compared to the offline algorithm that knows full future information, and present a critical tradeoff between battery capacity and system performance. Simulation results validate our analysis and demonstrate the superior performance of GLOBE compared to benchmark algorithms. \end{abstract}


%
\IEEEpeerreviewmaketitle
\section{Introduction}

Mobile computing and the Internet of Things are driving the development of many new
applications, turning data and information into actions that create new capabilities, richer experiences and unprecedented economic opportunities. Although cloud computing enables convenient access to a centralized pool of configurable and powerful computing resources, it often cannot meet the stringent requirements of latency-sensitive or geographically constrained applications, such as mobile gaming, augmented reality, tactile Internet and connected cars, due to the often unpredictable network latency and expensive bandwidth \cite{shi2016edge,mao2017mobile}. As a remedy to these limitations, Mobile Edge Computing (MEC) \cite{patel2014mobile} (a.k.a. fog computing \cite{chiang2016fog}) has recently emerged as a new computing paradigm to enable in-situ data processing at the Internet edge, in close proximity to mobile devices, sensors, actuators and connected things. In MEC, network edge devices, such as base stations (BSs) \cite{mao2017mobile}, are endowed with cloud-like functionalities to serve users' requests as a substitute of clouds, while significantly reducing the transmission latency as they are just one wireless hop away from end users and data sources.

In increasingly many scenarios, BSs are primarily powered by renewable green energy (e.g. solar and/or wind), rather than the conventional electric grid, due to various reasons such as location, reliability, carbon footprint and cost. The high intermittency and unpredictability of energy harvesting (EH) \cite{ulukus2015energy} significantly exacerbate the challenge of the latency requirements of applications as the computing capacity of an individual MEC-enabled BS is significantly limited at any moment of time. Geographical load balancing (GLB) \cite{lin2012online,islam2015water} is a promising technique for optimizing MEC performance by exploiting the spatial diversity of the available renewable energy to re-shape the computation workload distribution among the distributed BSs. However, energy harvesting leads to extraordinary challenges that existing GLB approaches cannot address: not only the available energy in the batteries imposes a stringent energy constraint at any time moment, but also the intrinsic evolution of these constraints couple the GLB decisions across time, and yet the decisions have to be made without foreseeing the future. Compared to existing GLB approaches for data center networks that solve time-decoupled problems, GLB for EH-powered MEC networks demands a fundamentally new design that can optimally manage limited energy, computing and radio access resources in both spatial and temporal domains.

\begin{figure}
  \centering
  \includegraphics[width=0.6 \linewidth]{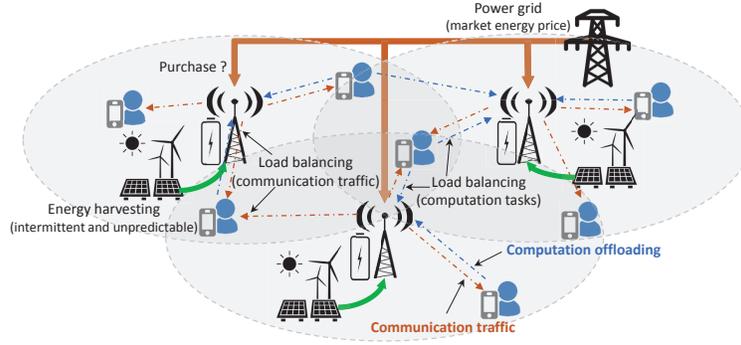}\\
  \caption{Illustration of load balancing in the MEC-enabled BS network with energy harvesting.}\label{sys_illu}
  \vspace{-20pt}
\end{figure}

In this paper, we study the joint problem of geographical load balancing, admission control and energy purchase among a network of EH-powered BSs (See Figure \ref{sys_illu} for an illustration) aimed at improving the mobile edge computing performance. The main contributions of this paper are as follows.

(1) We develop a novel framework, called GLOBE (Geographical LOad Balancing with Energy harvesting), for minimizing the long-term system cost (due to violating the computation delay constraint and dropping data traffic). In addition to the spatial coupling originated from load balancing among geographically distributed BSs, the considered problem also exhibits strong temporal coupling due to the energy harvesting causality constraint. Such \emph{spatial-temporal coupling} makes the investigated problem significantly different and more challenging than conventional GLB problems studied in the literature.

(2) We develop an online and distributed algorithm leveraging the Lyapunov stochastic optimization framework to solve the spatial-temporal GLB optimization problem. Instead of the conventional Lyapunov virtual queue technique, our algorithm is based on the Lyapunov perturbed queue technique to handle the energy causality constraint. The algorithm enables BSs to make GLB decisions without foreseeing the future system dynamics yet provides provable performance guarantee. In particular, we prove that GLOBE achieves within a bounded deviation from the optimal system performance that can be achieved by an oracle algorithm that knows the complete future information, provided that the battery capacity of the BSs is sufficiently large. The performance-battery capacity tradeoff is theoretically characterized.

(3) The perturbed battery queue technique decouples the GLB problem across time. In each time slot, we also develop a distributed algorithm to solve the per-time slot GLB problem. Lagrangian dual decomposition with quadratic regularization is used to relax the computation capacity constraint of individual BSs and decouple the spatially-coupled problem into subproblems that can be solved by individual BSs.

(4) We run extensive simulations to evaluate the performance of GLOBE and verify our analytical results for various system configurations and traffic arrival patterns. The results confirm that our method significantly improves the system performance compared to benchmark algorithms that do not perform GLB or only optimize the system myopically.

The rest of this paper is organized as follows. Section II discusses related works. Section III presents the system model and formulates the problem. Section IV develops the GLOBE framework and distributed and online solutions. Section V proves the performance guarantee of GLOBE. Section VI shows the simulation results. Finally, Section VII concludes the paper.

\section{Related Work}
Mobile Edge Computing \cite{mao2017mobile} (a.k.a fog computing computing \cite{chiang2016fog}) has emerged as a new paradigm to provide service with ultra-low latency and precise location awareness by bringing the computation resource closer to the end users. However, compared to the conventional mobile cloud computing, the edge servers are still limited in the computational and storage resource \cite{mach2017mobile} and hence may fail to offer the expected quality of service (QoS) when facing excessive workload. To address this challenge, hierarchical edge computing architectures \cite{tong2016hierarchical,xu2017online} are investigated to allow overloaded edge servers to further offloading computation tasks to the cloud server or the higher-tier servers. While this structure effectively solves the potential overloading problem at the edge servers, the computation tasks further offloaded to the cloud server will suffer from the large congestion delay due to backhaul transmission and hence degrades the QoS.


Many recent works \cite{chen2016socially, queis2015fogbalancing, xiao2017qoe} investigate cooperation among edge entities to improve uses' experience by fully utilizing the resource within the edge system and reducing the dependency on the cloud server. Although the cooperation among BSs in the conventional wireless communication has been well investigated subject to various constraints, e.g., radio resources \cite{abdelnasser2014clustering}, energy consumption budgets \cite{rubio2014association} and backhaul bandwidth capacity \cite{tam2017joint}), the cooperation among MEC-enabled BSs is a very different topic since the radio resource and computational resource need to be considered jointly. Efforts have been made to coordinate the edge entities to serve the mobile users collaboratively. In \cite{queis2015fogbalancing,queis2015small}, computation load distribution among the network of BSs is investigated by considering both radio and computational resource constraints, where clustering algorithms are proposed to maximize users' satisfaction ratio while keeping the communication power consumption low. In \cite{chen2016socially,tanzil2016distributed,Guruacharya2013Dynamic}, the coalitional game theory is applied to enable distributed formation of femto-clouds. However, these works study time-decoupled problems and do not consider energy-harvesting and hence, are very different from our paper. In \cite{gong2016networked}, cooperation among two BSs powered by renewable energy is investigated. However, it considers only two BSs and does not study MEC. Our prior work \cite{xu2017online} is one of the first works that consider energy harvesting in the context of MEC. However, it studies the hierarchical offloading problem for a single MEC-enable BS with energy harvesting capability but neglects the cooperation among multiple BSs. In this paper, we complete the story by exploring the cooperation among MEC-enabled BSs through geographical load balancing of communication traffic and computation workload.


Geographical Load Balancing (GLB) has been extensively studied in data center network (DCN) research \cite{lin2012online, xu2015temperature, lou2015spatio, liu2015greening, islam2015water}. Most of these works study load balancing problems that are independent across time and hence a myopic optimization problem is often formulated and solved to derive the GLB policy. Only a few works consider temporally coupled GLB problems. For example, authors in \cite{lin2012online} consider the temporal dependency with respect to the switching costs (turning on/off) of data center servers, which significantly differs from our considered problem where the temporal dependency is caused by the energy harvesting and energy consumption of BSs. Hence, different techniques are required to solve our problem. A similar temporally coupled GLB problem was studied under the framework of Lyapunov optimization in \cite{islam2015water}, where the authors consider a long-term water consumption constraint while optimizing the load balancing. However, firstly, the constraint in \cite{islam2015water} is a long-term average constraint whereas our paper considers a much more complicated and stringent battery causality constraint, and secondly, the constraint in \cite{islam2015water} is imposed on the entire network whereas in our paper each individual BS is constrained by its own (time-varying) battery state. To address these new challenges, we leverage the \emph{Lyapunov optimization with perturbation} technique \cite{neely2010dynamic} to develop our algorithm with a provable performance guarantee.

Recent works \cite{lin2012online,liu2015greening} studied the GLB with renewables, which show that GLB creates an important opportunity by allowing for ``follow the renewables'' routing. However, renewables in these works are considered as an instantaneous supplement to grid power, which needs to be matched by the energy demand instantaneously. By contrast, our considered problem uses renewables as the major power source and can be stored in the battery co-located with the BSs. Therefore, the battery state dynamics and the energy causality constraint need to be carefully considered while performing GLB. More importantly, the GLB algorithms (e.g. Averaging Fixed Horizon Control \cite{lin2012online}) in these works requires future information (albeit few), whereas our algorithm is able to work without foreseeing the future.

\section{System Model}
\subsection{Network Model}
We consider a network of $N$ base stations (BS), indexed by $\mathcal{N} = \{1,2,...,N\}$, providing communication and edge computing services to users. Each BS is mainly powered by renewable energy harvested from wind and/or solar radiation, and is equipped with a battery for energy storage. Renewable energy is considered as free. However, BSs can also purchase energy from the electric grid whenever needed, at a cost depending on the current energy market price and their energy demand. We consider a dense deployment scenario where BSs have overlapping coverage areas so that geographical load balancing is possible.

Time is discretized, with each time slot matching the timescale at which load balancing decisions can be updated. In each time slot $t$, each BS $i$ has a set of associated users $\mathcal{U}_i^t$ according to a certain user-cell association rule (e.g., maximizing the received signal strength). The BS can collect information (e.g. demand, channel conditions etc.) of its associated users and dispatch their loads to itself or nearby BSs thanks to the dense deployment of BSs. Having the associated BSs to make load balancing decisions instead of the user themselves can significantly reduce the algorithm complexity and information exchange overhead, thereby improving the scalability of the system. Let $\mathcal{U}^t = \cup_{i\in\mathcal{N}} \mathcal{U}_{i}^t$ be the set of all users in time $t$. For each user $u$, let $\mathcal{N}_u \subseteq \mathcal{N}$ be the set of BSs that are in its transmission range.

Although our framework and algorithm can handle a time-varying user set across time slots, for the ease of exposition, we will assume that the user set is static over the considered time horizon. In an alternative network model, we can also imagine that each user represents a service area which remains static although users inside it can change. Therefore, in the remainder of this paper, we drop the time index $t$ in the user set $\mathcal{U}^t$ and $\mathcal{U}^t_i, \forall i\in \mathcal{N}$.

\subsection{Communication Traffic and Cost Model}
Users have uplink and downlink communication traffic that has to be served by the BSs. We assume that uplink and downlink transmissions operate on orthogonal channels and focus on the downlink traffic since energy consumption associated with downlink transmissions predominates the total BS energy consumption \cite{Auer2011}. Henceforth, communication traffic refers to downlink communication traffic. In each time slot $t$, each user $u$ has communication traffic at arrival rate $\mu^t_i \leq \mu^{max}$, which we assume to be upper bounded by $\mu^{max}$. The data size of communication traffic arrival is modeled as an exponential random variable with mean $\omega_u$ for user $u$. Using Shannon's theorem, the expected energy consumption for each transmission by BS $i$ to its associated user $u$ can be computed as
\begin{align}
p^t_{i,u} = \mathbb{E}\left[\frac{P_{tx, i}\omega_u}{W\log_2\left(1 + \frac{H^t_{i,u}P_{tx,i}}{\sigma^2}\right)}\right],
\end{align}
where $W$ is the downlink bandwidth, $H^t_{i,u}$ is a random variable representing the downlink channel gain between BS $i$ and user $u$ in time slot $t$, $P_{tx,i}$ is BS $i$'s transmitting power, $\sigma^2$ is the noise power, and the expectation is taken over the data size and the channel state.

For each user $u$, its communication traffic can be transmitted by any of the BSs in $\mathcal{N}_u$ that covers $u$, e.g., via Coordinated Multi Point (CoMP) with Coordinated Scheduling (CS) \cite{Lee2012}. Denote $\alpha^t_u = \{\alpha^t_{u,j}\}_{j\in\mathcal{N}_u}$ as user $u$'s transmission strategy where $\alpha^t_{u,j}$ represents the amount of traffic that is transmitted by BS $j$. Since a BS has multiple associated users, the transmission load balancing strategy of BS $i$ is the collection of individual users' transmission strategy, represented by $\bm\alpha^t_i = \{\bm\alpha^t_u\}_{u\in\mathcal{U}_i}$. Further, the transmission load balancing strategy of the whole network is collected in the notation $\bm\alpha^t = \{\bm\alpha^t_i\}_{i\in\mathcal{N}}$. Clearly, a feasible transmission load balancing strategy must satisfy
\begin{align}\label{demand2}
\sum_{j\in\mathcal{N}_u}\alpha^t_{u,j} \leq \mu^t_u, \forall u.
\end{align}
Note that the strict inequality in the above constraint means that some communication traffic of a user may not be fulfilled by any BS. This could happen when energy supply of the BSs becomes a concern and hence the BSs have to drop some communication traffic to save energy. An alternative interpretation of constraint \eqref{demand2} is that user requested files, especially video files, are transmitted at a lower resolution and hence only the incremental high-resolution layers of the video will be dropped.

Depending on the amount of transmitted traffic and the current channel conditions, the transmission energy consumption of BS $i$ in time slot $t$ is
\begin{align}
E_{tx,i}(\bm\alpha^t) = \sum_{u \in \mathcal{U}}p^t_{i,u}\alpha^t_{u,i}.
\end{align}
Dropping communication traffic incurs costs to the network operator due to, e.g., user dis-satisfaction or reduced Quality of Experience (QoE). The cost is user-specific depending on user priority. The total communication traffic dropping cost of BS $i$ is thus
\begin{align}
C_{tx, i}(\bm\alpha^t) = \sum_{u\in\mathcal{U}_i}c_{tx, u}(\mu^t_u - \sum_{j\in\mathcal{N}_u}\alpha^t_{u,j}),
\end{align}
where $c_{tx,u}$ is the unit dropping cost for user $u$ which converts the dropped communication traffic into a monetary value.

\subsection{Computation Tasks and Cost Model}
In each time slot $t$, each user also generates computationally intensive tasks that have to be offloaded to the BS/edge server for processing. When the computation is finished, the computation results will be returned to the user. The computation task arrival of user $u$ is denoted by $\lambda^t_u \leq \lambda^{max}$, where $\lambda^{max}$ is the maximum arrival rate. For each computation task, the required number of CPU cycles is assumed to be an exponential random variable with mean $\rho$. The computation capability of BS $i$ is measured by its CPU speed (i.e. CPU cycles per second), denoted by $f_i$. Given the the constant CPU speed of BS $i$, and the exponential distribution of workload of each computation task, the processing time of a computation task follows an exponential distribution with mean $\rho/f_i$. Therefore, if BS $i$ processes computation tasks with arrival rate $\lambda$, the average computation delay (including the waiting time and the processing time) for a task, can be computed using a M/M/1 queuing model
\begin{align}
d_i^t = \frac{1}{f_i/\rho - \lambda}.
\end{align}
The workload is delay-sensitive which has to satisfy a maximum delay constraint $d^{max}$. If the total computation workload on a BS is too large, then the delay constraint will be easily violated. Therefore, computation load balancing is also performed to exploit the under-used, otherwise wasted, computational resources on BSs with light workload to improve the overall system performance.

In dense networks, computation load balancing is realized by offloading computation tasks to nearby BSs other than the directly associated BS. For each user $u$, let $\bm\beta^t_u =\{\beta^t_{u, j}\}_{j\in\mathcal{N}_u}$ be its offloading strategy. A computation load balancing strategy of BS $i$ is the collection of its users' offloading strategy $\bm\beta^t_i = \{\bm\beta^t_u\}_{u\in\mathcal{U}_i}$ and a computation load balancing strategy of the whole network is $\bm\beta^t = \{\bm\beta^t_i\}_{i\in\mathcal{N}}$. In a feasible computation load balancing strategy, the total offloaded workload of user $u$ must not exceed its computation demand, namely
\begin{align}\label{demand1}
\sum_{j\in\mathcal{N}_u} \beta^t_{u,j} \leq \lambda^t_u, \forall u.
\end{align}
Moreover, the received workload of BS $i$ must not exceed its computation capacity to satisfy the delay constraint, namely
\begin{align}\label{capacity}
\sum_{u \in \mathcal{U}} \beta^t_{u,i} \leq f_i/\rho - 1/d^{max}, \forall i.
\end{align}
Due to the above constraints, users can also choose to drop some computation workload (or process locally on user devices, or offload to the remote cloud) if the edge BS network collectively cannot support it.

The computation energy consumption of BS $i$ is proportional to its received workload (from its own users or users of nearby BSs), namely $\sum_{u \in \mathcal{U}} \beta^t_{u,i}$, and the square of its CPU speed $(f_i)^2$ \cite{mao2017mobile}. Therefore, the computation energy consumption of BS $i$ in time slot $t$ is
\begin{align}
E_{com,i}(\bm\beta^t) = \kappa(f_i)^2 \sum_{u \in \mathcal{U}} \beta^t_{u,i},
\end{align}
where the coefficient $\kappa$ is for unit task energy consumption. Dropping computation workload incurs a cost, which is linear to the dropped computation workload. The computation dropping cost for BS $i$ (considering its own associated users) is therefore
\begin{align}
C_{com, i}(\bm\beta^t) = \sum_{u\in\mathcal{U}_i}c_{com,u}(\lambda^t_u - \sum_{j\in\mathcal{N}_u}\beta^t_{u,j}),
\end{align}
where $c_{com,i}$ is the unit workload dropping cost for user $u$ which converts the dropped computation tasks into a monetary value.

\subsection{Energy Harvesting and Purchase}
BSs in the considered system are mainly powered by renewable energy harvested from the environment, such as wind energy and/or solar energy.  To capture the intermittent and unpredictable nature of the energy harvesting process, we model it as successive energy packet arrivals, i.e. in each time slot $t$, energy packets with amount $\mathcal{E}_i^t \leq \mathcal{E}^{max}$ arrive at BS $i$, where $\mathcal{E}_i^t$ is drawn from some unknown distribution upper bounded by $\mathcal{E}^{max}$. In each time slot $t$, part of the arrived energy, denoted by $e^t_i$, satisfying
\begin{align}\label{EH}
0 \leq e^t_i \leq \mathcal{E}_i^t, \forall i
\end{align}
will be harvested and stored in the battery of BS $i$. We start with assuming that the battery capacity is sufficiently large. Later we will show that by picking the value of $e^t_i$, the battery energy levels can be deterministically upper-bounded under the proposed algorithm, thus we only need finite-capacity batteries in the actual implementation. More importantly, including $e^t_i$'s as decision variables in the optimization facilitates the derivation and performance analysis of the proposed algorithm. Similar techniques were adopted in existing works \cite{lakshminarayana2014cooperation,mao2015lyapunov}. We collect the energy harvesting decisions of all BSs in the notation $\e^t = \{e^t_1,...,e^t_N\}$.

When renewable energy falls short, BSs can also purchase energy from the electric grid depending on its energy demand and the current market energy price, which varies over time. Let $c_{grid}^t$ be the unit energy price in time slot $t$. BS $i$'s energy purchase decision is denoted by $g^t_i \in [0, g^{max}]$ where $g^{max}$ is the maximum grid energy a BS can purchase. The energy purchase cost is thus
\begin{align}
C_{grid,i}(g^t_i) = c_{grid}^t g^t_i.
\end{align}

\subsection{Problem Formulation}
The objective of the network operator is to minimize the total system cost in \eqref{eq:total_sys_cost} due to dropping communication traffic/computation workload and purchasing grid energy by jointly optimizing load balancing, energy harvesting and purchase in each time slot.
\begin{align}\label{eq:total_sys_cost}
	C_i(\bm\alpha^t, \bm\beta^t, g^t_i) \triangleq C_{tx, i}(\bm\alpha^t) + C_{com, i}(\bm\beta^t) + C_{grid,i}(g^t_i)
\end{align}
 The optimization is subject to an energy availability constraint, namely the consumed energy in each time slot must not exceed what is available. The total energy consumption of BS $i$ in time slot $t$ includes communication and computation energy consumption, i.e.
\begin{align}
E_i^t(\bm\alpha^t, \bm{\beta}^t) = E_{tx, i}^t(\bm{\alpha}^t) + E_{com, i}^t(\bm{\beta}^t)
\end{align}
Let $B^t_i$ denote the available battery energy at the beginning of time slot $t$ for BS $i$. The energy causality constraint must be satisfied in every time slot
\begin{align}\label{causality}
E_i^t(\bm\alpha^t, \bm{\beta}^t) \leq B^t_i, \forall i \in \mathcal{N}
\end{align}
At the end of each time slot, the battery state evolves as follows
\begin{align}
B^{t+1}_i = \min\{B^t_i - E_i^t(\bm\alpha^t, \bm{\beta}^t) + e^t_i + g^t_i, B^{max}\} \label{dynamics}
\end{align}
where $B^{max}$ is the battery capacity.

The problem can thus be formulated as follows:
\begin{subequations}
	\begin{align}
	\textbf{P1}~~~\min_{{\small\bm{\alpha}^t, \bm{\beta}^t}, \e^t, \g^t, \forall t}~~&\lim_{T\to\infty}\frac{1}{T}\sum_{t=1}^T\sum_{i\in\mathcal{N}} \mathbb{E}\left[C_i(\bm\alpha^t, \bm\beta^t, g^t_i)\right]\\
	\text{subject to}~~~ 
	&\text{Constraints \eqref{demand2},  \eqref{demand1}, \eqref{capacity}, \eqref{EH}, \eqref{causality}}, \forall t \nonumber
	\end{align}
\end{subequations}
where \eqref{demand2}, \eqref{demand1} are the constraints for feasible transmission load balancing strategy and feasible computation load balancing strategy, respectively. \eqref{capacity} is the computation capacity constraint for each BS. \eqref{EH} indicates that the harvested energy cannot exceed the green energy arrivals. \eqref{causality} is the battery causality constraint, indicating that the energy consumption in the current time slot should not exceed the battery level.

Because of the battery state dynamics and energy causality constraints, the load balancing decisions are highly coupled across time slots. Let $C^*_1$ be the infimum time average system cost achievable by any policy that meets the required constraints in every time slot, possibly by an oracle algorithm that has complete future information of the communication traffic arrival process, the computation task arrival process, the energy harvesting process, the market energy price, and the channel conditions. We note that $C^*_1$ represents a performance upper bound for practical algorithms as they do not possess complete and accurate future information. In the next sections, we will develop an algorithm that achieves $C^*_1$ within a bounded deviation without requiring future information.

\section{Online Load Balancing for EH-Powered MEC}
\subsection{An Online Algorithm based on Perturbed Lyapunov Optimization}
Lyapunov optimization is a powerful framework that enables online stochastic optimization without requiring future system dynamics yet provides provable performance guarantee. However, conventional Lyapunov optimization \cite{neely2010stochastic} is not directly applicable for solving $\textbf{P1}$ due to the presence of energy causality constraints \eqref{causality}. In order to circumvent this obstacle, we take an alternative approach based on the perturbed Lyapunov technique similar to \cite{urgaonkar2011optimal}. First, we formulate a slightly modified version of $\textbf{P1}$  as follows:
\begin{subequations}
	\begin{align}
	\textbf{P2}~~~\min_{\bm{\alpha}^t, \bm{\beta}^t, \e^t, \g^t, \forall t} &~~~ \lim_{T\to\infty}\frac{1}{T}\sum_{t=1}^T\sum_{i\in\mathcal{N}} \mathbb{E}\left[C\left(\bm\alpha^t, \bm\beta^t, g^t_i\right)\right] \\
	\text{subject to} & ~~~ \lim_{T\to\infty}\frac{1}{T}\sum_{t=1}^T\sum_{i\in\mathcal{N}}\mathbb{E} \left[E_i^t(\bm\alpha^t, \bm{\beta}^t) - g^t_i - e^t_i\right]=0\label{relaxed_causality} \\
	&~~~\text{Constraints \eqref{demand2}, \eqref{demand1}, \eqref{capacity},  \eqref{EH}}, \forall t \nonumber
	\end{align}
\end{subequations}
where the energy causality constraint \eqref{causality} in \textbf{P1} is replaced with a long-term energy demand and supply clearance constraint \eqref{relaxed_causality}. It can be easily shown that \textbf{P2} is a relaxed version of \textbf{P1}, as any feasible solution to \textbf{P1} would also satisfy the constraints in \textbf{P2}. To see this, consider any policy that satisfies \eqref{EH} and \eqref{causality}, then summing  equation \eqref{dynamics} over $t \in \{1,\dots,T\}$, dividing by $T$ and taking limits as $T\rightarrow\infty$  yields \eqref{relaxed_causality}. Let $C^*_2$ denote the optimal value of \textbf{P2} and then, we must have $C^*_2\leq C^*_1$. Following the framework in \cite{neely2010stochastic}, it can be shown that the optimal solution to the relaxed problem \textbf{P2} can be obtained by the method of stationary randomized policy, stated in the following lemma.
\begin{lemma}
There exists a stationary and possibly randomized policy $\Pi$ that achieves
\begin{align}
\sum_{i\in\mathcal{N}} \mathbb{E}\left[C_i(\bm\alpha^{\Pi,t}, \bm\beta^{\Pi,t}, g^{\Pi,t}_i)\right]=C^*_2 \label{lemmaC}
\end{align}
while satisfies the constraints \eqref{demand1}, \eqref{capacity}, \eqref{demand2}, \eqref{EH} in \textbf{P2} and
\begin{align}
\mathbb{E}\left[E_i^t(\bm\alpha^{\Pi,t}, \bm{\beta}^{\Pi,t}) - g^{\Pi,t}_i - e^{\Pi,t}_i\right]=0 \label{lemmaE}
\end{align}
\end{lemma}
\begin{proof}
The proof follows the framework in \cite{neely2010stochastic} and is omitted here for brevity.
\end{proof}
However, the optimal policy $\Pi$ requires the knowledge of the system dynamics over the entire time horizon, and hence is very difficult to derive in practice if not impossible. Next, we develop an online algorithm, called GLOBE, to solve \textbf{P1} and compare its performance with $C^*_2$ achieved by $\Pi$. We first define the \textit{perturbed battery queue} for each BS.
\begin{definition}
	The perturbed battery queue $\tilde{B}^t_i$ of BS $i$ is defined as
	\begin{align}
	\tilde{B}^t_i=B^t_i-\theta, \forall i\in\mathcal{N}
	\end{align}
	where $\theta$ is the perturbation parameter.
\end{definition}
The value of $\theta$ will be specified later when we analyze the algorithm performance. In a nutshell, the proposed GLOBE algorithm aims to minimize the weighted sum of the system cost and the perturbed energy queue in each time slot, which shall stabilize $B^t_i$ around the perturbed energy level $\theta$ and meanwhile minimize the system cost, where the weight is adaptively updated over time. In each time slot $t$, the load balancing strategies, the amount of harvested and purchased energy are determined by solving the following optimization problem:
\begin{subequations}
	\begin{align}
	\textbf{P3}~~~\min_{\bm\beta^t, \bm\alpha^t, \e^t, \g^t} &~~~ \sum_{i\in\mathcal{N}} \left(V\cdot C_i(\bm\alpha^t, \bm\beta^t, g^t_i) - \tilde{B}^t_i \cdot (E_i^t(\bm\alpha^t, \bm{\beta}^t) - g^t_i - e^t_i) \right)\\
	\text{subject to}&~~~\text{Constraints \eqref{demand2}, \eqref{demand1}, \eqref{capacity}, \eqref{EH}}
	\end{align}
\end{subequations}
which is parameterized by only the current system state (i.e., communication traffic arrival, computation task arrival, energy price, channel conditions in time slot $t$, etc). Therefore, our algorithm can work online without requiring future information of the system dynamics. Fig. \ref{fig:globeillu} illustrates GLOBE for an arbitrary time slot $t$. At the beginning of each time slot, GLOBE observes the system states and battery level. Then, with the observed information, it derives the optimal energy harvesting, grid power purchase, and load balancing decisions. At the end of the time slot, the battery states of the BSs are updated based on the acquired energy (via harvesting and purchase) and the consumed energy (due to transmission and computation), thereby linking per-time slot problems across time.
\begin{figure}[htb]
	\centering
	\includegraphics[width=0.65\linewidth]{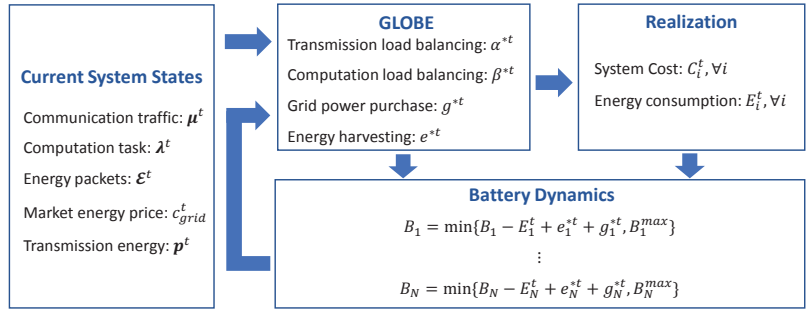}
	\caption{Illustration of GLOBE for time slot $t$.}
	\label{fig:globeillu}
\vspace{-20pt}
\end{figure}

\begin{algorithm}[htb]
	\caption{The GLOBE algorithm}
	\begin{algorithmic}[1]
		\State \textbf{Input}: $\{\theta_i,  i \in \mathcal{N}\}$, $V$
		\State \textbf{Output}: Load balancing $\bm\alpha^t$ and $\bm\beta^t$, purchased energy $\g^t$, harvested energy $\e^t$, for every $t$
		\For{every time slot $t$}
		\State Observe $\bm\mu^t, \bm\lambda^t, \mathcal{E}^t, c^t_{grid}, \p^t$
		\State Solve (\textbf{P3}) to get $\bm\alpha^t$, $\bm\beta^t$, $\e^t$ and $\g^t$
		\State Update the battery state according to \eqref{dynamics} for $i \in \mathcal{N}$
		\EndFor
	\end{algorithmic}
\end{algorithm}

The GLOBE algorithm is formally presented in Algorithm 1. How this algorithm is derived will be explained and its performance will be theoretically analyzed in Section V. Now, to complete GLOBE, it remains to solve the per-time slot problem \textbf{P3}.

The objective function in \textbf{P3} can be decomposed into three parts:
\begin{equation}\label{p3_decom}
\begin{split}
&\sum_{i\in\mathcal{N}}\left(V\cdot C_i(\bm\alpha^t, \bm\beta^t, g^t_i) - \tilde{B}^t_i \cdot(E_i^t(\bm\alpha^t, \bm{\beta}^t) - g^t_i - e^t_i)\right)\\
=&\sum_{i\in\mathcal{N}}\left(V\cdot c^t_{grid}g^t_i + \tilde{B}^t_i (g^t_i + e^t_i) \right)\\
+ &\sum_{i\in\mathcal{N}}\left(V\cdot \sum_{u\in\mathcal{U}_i}c_{tx,u}(\mu^t_u - \sum_{j\in\mathcal{N}_u}\alpha^t_{u,j}) - \tilde{B}^t_i \sum_u p^t_{i,u}\alpha^t_{u,i}\right) \\
+&\sum_{i\in\mathcal{N}} \left( V\cdot \sum_{u\in \mathcal{U}_i}c_{com,u}(\lambda^t_u - \sum_{j\in\mathcal{N}_u}\beta^t_{u,j}) - \tilde{B}^t_i \kappa(f_i)^2 \sum_{u\in\mathcal{U}}\beta^t_{u,i}\right).
\end{split}
\end{equation}
Given the current  perturbed battery queues of all BSs, the first term on the right-hand side depends only on the transmission load balancing strategy $\bm\alpha^t$, the second term depends only on the computation offloading load balancing strategy $\bm\beta^t$, and the third term depends only on the energy harvesting $\e^t$ and purchase strategies $\g^t$. Therefore, we can optimize each term separately.

\vspace{-10pt}
\subsection{Optimal Energy Harvesting and Purchase}
We start with solving the per-slot optimal energy harvesting and purchase problem, which is the first term of \eqref{p3_decom}:
\begin{subequations}
	\begin{align}
	\min_{\e^t, \g^t}~~~\sum_{i\in\mathcal{N}} \left(V\cdot c^t_{grid}g^t_i + \tilde{B}^t_i (g^t_i + e^t_i) \right)~~~~
	\text{subject to}~~~0 \leq e^t_i \leq \mathcal{E}_i^t, \forall i.
	\end{align}
\end{subequations}
This is a simple linear programming (LP) and can be solved by each BS in a fully distributed way without requiring any message exchange among the BSs. The optimal solution can be derived as follows
\begin{align}
e^{t*}_i = \mathcal{E}^t_i \cdot \mathbf{1}\{\tilde{B}^t_i \leq 0\}, \forall i;~~~~
g^{t*}_i = g^{max} \cdot \mathbf{1}\{V\cdot c^t_{grid} + \tilde{B}^t_i \leq 0\}, \forall i
\end{align}
where $\mathbf{1}\{\cdot\}$ is the indicator function. The optimal solution suggests that
\begin{enumerate}
\item BS $i$ should harvest all possible energy if the battery queue $B^t_i$ is less than a threshold $\theta_i$ and zero energy otherwise;
\item BS $i$ should purchase the maximum amount of grid power if the battery queue is less than a threshold $\theta_i - V\cdot c^t_{grid}$ and zero grid energy otherwise.
\end{enumerate}

Since $\theta_i - V\cdot c^t_{grid} < \theta_i$, renewable energy has a high priority than the grid power, i.e.,  renewable energy will be acquired before purchasing any grid energy. This is intuitive because renewable energy is free. Moreover, because the grid energy price $c^t_{grid}$ varies over time, the energy purchase strategy also varies depending on $c^t_{grid}$. The BS is more likely to purchase grid energy when the price is low. The energy harvesting and purchase strategies are very simple and both have a ``none-or-all'' structure. However, we will show later that these simple strategies indeed ensure that we only need a finite battery capacity for each BS while providing provable performance guarantee.

\subsection{Optimal Transmission Load Balancing}
The optimal transmission load balancing can be obtained by solving the following LP problem:
\begin{subequations}
	\begin{align}
	\max_{\bm\alpha^t}&~~~\sum_{i\in\mathcal{N}} \left( \sum_{u\in\mathcal{U}_i}Vc_{tx,u}\sum_{j\in\mathcal{N}_u}\alpha^t_{u,j} + \tilde{B}^t_{i}\sum_{u} p^t_{i,u}\alpha^t_{u,i}\right)\label{TLB}\\
	\text{subject to}&~~~\sum_{j\in\mathcal{N}_u}\alpha^t_{u,j} \leq \mu^t_{u}, \forall i, \forall u\in\mathcal{U}_i
	\end{align}
\end{subequations}
The objective function can be further rearranged as follows
\begin{align*}
\sum_{i\in\mathcal{N}} \left( \sum_{u\in\mathcal{U}_i}Vc_{tx,u}\sum_{j\in\mathcal{N}_u}\alpha^t_{u,j} + \tilde{B}^t_{i}\sum_{u} p^t_{i,u}\alpha^t_{u,i}\right)
=\sum_{i\in\mathcal{N}} \sum_{u\in\mathcal{U}_i} \sum_{j\in\mathcal{N}_u}(Vc_{tx,u} + \tilde{B}^t_j p^t_{j,u})\alpha^t_{u,j}.
\end{align*}
This enables each BS $i$ to decide its transmission load balancing strategy for its own associated users in a distributed way, after exchanging information of the perturbed battery queue lengths with neighbor BSs. Specifically, BS $i$ solves the following problem for each of its user $u\in\mathcal{U}_i$:
\begin{subequations}
	\begin{align}
\max_{\bm\alpha^t_u}~~\sum_{j\in\mathcal{N}_u}(Vc_{tx,u}+\tilde{B}^t_jp^t_{j,u})\alpha^t_{u,j} ~~~~\text{subject to}~~\sum_{j\in\mathcal{N}_u}\alpha^t_{u,j} \leq \mu^t_{u}. \label{TLB-i}
    \end{align}
\end{subequations}

The optimal solution has a simple structure. If there exist some BS $j$ such that the coefficient $Vc_{tx,u} + \tilde{B}^t_j p^t_{j,u} \geq 0$, then all user $u$'s communication traffic will solely be transmitted by the BS that has the largest positive coefficient. If no BS has positive coefficients, then user $u$'s downlink traffic will not be transmitted, thereby causing communication traffic to be dropped. We note that whether traffic dropping occurs or not depends on the cost $c_{tx,u}$. When $c_{tx,u}$ is sufficiently large (i.e. dropping communication traffic is very costly), then there will always be BSs with positive coefficients and hence no traffic will be dropped.

\subsection{Optimal Computation Load Balancing}
The optimal computation load balancing strategy can be obtained by solving the following LP problem
\begin{subequations}
	\begin{align}
	\max_{\bm\beta^t} &~~~\sum_{i\in\mathcal{N}} \left( \sum_{u\in\mathcal{U}_i}Vc_{com,u}\sum_{j\in\mathcal{N}_u}\beta^t_{u,j} + \tilde{B}^t_{i}\sum_{u} \kappa(f_i)^2\beta^t_{u,i}\right) \label{CLB}\\
	\text{subject to}&~~~\sum_{j\in\mathcal{N}_u}\beta^t_{u,j} \leq \lambda^t_{u}, \forall i, \forall u\in\mathcal{U}_i\\
	&~~~\sum_{u\in\mathcal{U}}\beta^t_{u,i} \leq f_i/\rho - 1/d^{max}, \forall i  \label{decoupled constraints}
	\end{align}
\end{subequations}

Similar to the transmission load balancing problem, the objective function can also be rearranged as follows:
\begin{align*}
\sum_{i\in\mathcal{N}}\left(\sum_{u\in\mathcal{U}_i}Vc_{com,u}\sum_{j\in\mathcal{N}_u}\beta^t_{u,j} + \tilde{B}^t_i \kappa (f_i)^2\sum_u \alpha^t_{u,i} \right)
= \sum_{i\in\mathcal{N}}\sum_{u\in\mathcal{U}_i}\sum_{j\in\mathcal{N}_u}(Vc_{com,u} +\tilde{B}^t_j\kappa(f_j)^2)\beta^t_{u,j}.
\end{align*}
While the computation load balancing problem looks similar to the transmission load balancing problem in the previous subsection, it indeed is a much more complicated problem because of the computation capacity constraint \eqref{decoupled constraints} that makes the load balancing decisions among BSs highly coupled. To enable a distributed solution, we use Lagrangian dual decomposition to relax the computation capacity constraint and decouple the primal problem into several subproblems. Unfortunately, the above considered LP problem is sensitive to perturbations and does not work well with the dual decomposition technique. To overcome this difficulty, we add a quadratic regularization term to the objective function to smooth the original LP problem{\footnote{We note that a similar approach was proposed and adopted in \cite{yuan2013linearized}.}. The regularized problem is:
\begin{subequations}
	\begin{align}
	\max_{\bm\beta^t}&~~~\sum_{i\in\mathcal{N}}\sum_{u\in\mathcal{U}_i}\sum_{j\in\mathcal{N}_u}\left((Vc_{com,u} +\tilde{B}^t_j\kappa(f_j)^2)\beta^t_{u,j} - \frac{1}{2\epsilon}(\beta^t_{u,j})^2 \right) \label{CLB regularized}\\
	\text{subject to}&~~~\sum_{j\in\mathcal{N}_u}\beta^t_{u,j} \leq \lambda^t_{u}, \forall i, \forall u\in\mathcal{U}_i;~~~~\sum_{u\in\mathcal{U}}\beta^t_{u,i} \leq f_i/\rho - 1/d^{max}, \forall i
	\end{align}
\end{subequations}
According to \cite{friedlander2007exact}, there exists $\epsilon_{min} > 0$ such that the optimal solution to problem (\ref{CLB regularized}) is also an optimal solution to problem (\ref{CLB}) for any $\epsilon \geq \epsilon_{min}$. For our problem, we have  found that $\epsilon = 10^7$ works well. Therefore, we will instead solve the smoothed quadratic programming (QP) problem. We associate a Lagrangian multiplier $\gamma_i$ with the computation capacity constraint of each BS $i$. The Lagrangian of the primal problem \eqref{CLB regularized} is thus
\begin{align*}
L(\bm\gamma, \bm\beta^t) = &\sum_{i\in\mathcal{N}}\sum_{u\in\mathcal{U}_i}\sum_{j\in\mathcal{N}_u}(Vc_{com,u} +\tilde{B}^t_j\kappa(f_j)^2)\beta^t_{u,j} - \frac{1}{2\epsilon}(\beta^t_{u,j})^2 - \sum_{i\in\mathcal{N}}\gamma_i(\sum_{u\in\mathcal{U}}\beta^t_{u,i} - f_i/\rho + 1/d^{max})\nonumber\\
=&\sum_{i\in\mathcal{N}} \sum_{u\in\mathcal{U}_i}\sum_{j\in\mathcal{N}_u}(Vc_{com,u} + \tilde{B}^t_j \kappa (f_j)^2 - \gamma_j)\beta^t_{u,j} - \frac{1}{2\epsilon}(\beta^t_{u,j})^2 + \sum_{i\in\mathcal{N}}\gamma_i(f_i/\rho - 1/d^{max}).   \label{Lagrangian}
\end{align*}
Because the third term $\sum_{i\in\mathcal{N}}\gamma_i(f_i/\rho - 1/d^{max})
$ is independent of the computation load balancing strategy $\bm\beta^t$, the relaxed problem is
\begin{subequations}
	\begin{align}
	\max_{\bm\beta^t}&~~~ \sum_{i\in\mathcal{N}} \sum_{u\in\mathcal{U}_i}\sum_{j\in\mathcal{N}_u}(Vc_{com,u} + \tilde{B}^t_j \kappa (f_j)^2 - \gamma_j)\beta^t_{u,j} - \frac{1}{2\epsilon}(\beta^t_{u,j})^2\\
	\text{subject to}&~~~\sum_{j\in\mathcal{N}_u}\beta^t_{u,j} \leq \lambda^t_{u}, \forall i, \forall u\in\mathcal{U}_i
	\end{align}
\end{subequations}

After this relaxation, the problem can be divided into two levels. At the lower level, each BS seeks to optimize $\bm\beta^t$ for given $\bm\gamma$. The subproblem for BS $i$ is
\begin{subequations}
	\begin{align}
	\max_{\bm\beta^t_i}&~~~ \sum_{u\in\mathcal{U}_i}\sum_{j\in\mathcal{N}_u}(Vc_{com,u} + \tilde{B}^t_j \kappa (f_j)^2 - \gamma_j)\beta^t_{u,j} - \frac{1}{2\epsilon}(\beta^t_{u,j})^2 \label{lower}\\
	\text{subject to}&~~~\sum_{j\in\mathcal{N}_u}\beta^t_{u,j} \leq \lambda^t_{u}, \forall u\in\mathcal{U}_i
	\end{align}
\end{subequations}
The above problem is a QP. Its optimal value is denoted as $L_i(\bm\gamma)$. Since the problem is strictly concave, the optimal solution is unique.

At the higher level, each BS $i$ seeks to solve the dual problem:
\begin{align}
\min_{\bm\gamma \geq 0}&~~\sum_{i\in\mathcal{N}}(L_i(\bm\gamma) + \gamma_i(f_i/\rho - 1/d^{max}). \label{higher}
\end{align}
Since the lower level problem has a unique optimal solution, the higher level dual problem is differentiable. Therefore, we use the gradient descent method to solve this problem in an iterative manner. In each iteration $k$, the BSs exchange the current values of $\bm\gamma^{(k)}$ and each BS $i$ solves the lower level problem \eqref{lower} using the current iteration $\bm\gamma^{(k)}$ to obtain the optimal $\bm\beta^{t, (k)}_i$. Then the BSs exchange the derived $\bm\beta^{t, (k)}_i$ and use it to update the Lagrangian $\bm\gamma^{(k+1)}$ for the next iteration as follows:
\begin{align}
\gamma_i^{(k+1)} = [\gamma_i^{(k)} - \delta^{(k)}(f_i/\rho - 1/d^{max} - \sum_{u}\beta^{t,(k)}_{u,i}(\bm\gamma^{(k)}))]^+, \label{update}
\end{align}
where $\delta^{(k)} > 0$ is the step size of the subgradient descent method. The dual variable will converge to the optimal $\bm\gamma^*$ as $k \to \infty$. Fig. \ref{fig:distralgillu} illustrates the procedure of the iterative algorithm for one iteration. Since the Slater's condition (i.e., the primal problem is convex and there exists a feasible solution) is satisfied for the primal problem, the strong duality holds. As a result, the optimal solution $\bm\beta^t$ for the primal problem can also be derived through $\bm\beta^{t, (k)}(\bm\gamma^{(k)})$ with $k \to \infty$. We therefore have solved the primal problem. Algorithm 2 summarizes the iterative algorithm which solves the computation load balancing problem in a distributed manner.

\begin{figure}[t]
	\centering
	\includegraphics[width=1 \linewidth]{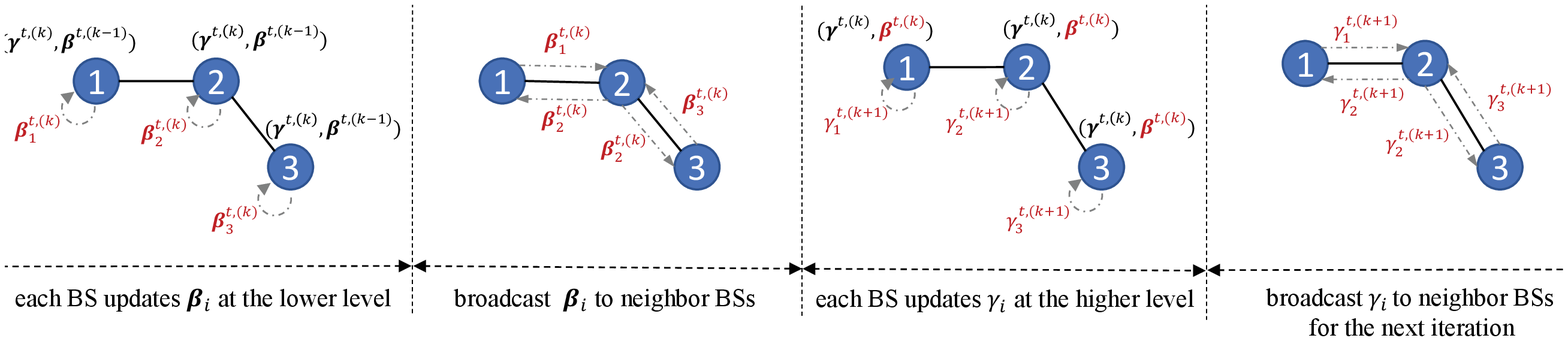}
	\caption{Illustration of the iterative algorithm for one iteration.}
	\label{fig:distralgillu}
\vspace{-20pt}
\end{figure}

\begin{algorithm}[htb]
	\caption{The Distributed Algorithm That Iteratively Solves Per-Slot Computation Load Balancing}
	\begin{algorithmic}[1]
		\State \textbf{Input}: $\tilde{B}^t_i$ for each BS $i$ at the current time $t$
		\State \textbf{Output}: Slot $t$'s computation load balancing $\bm\beta^t$
        \State Each BS $i$ broadcast $\tilde{B}^t_i$ to its neighbor BSs.
		\For{iteration $k$}
            \State  BS $i$ broadcasts its $\gamma^{(k)}_i$ to its neighbor BSs, $\forall i \in \mathcal{N}$
            \State  BS $i$ solves \eqref{lower} for $\bm\beta^{t, (k)}_i$, $\forall i \in \mathcal{N}$
            \State  BS $i$ broadcasts its solution $\bm\beta^{t, (k)}_i$ to its neighbor BSs, $\forall i \in \mathcal{N}$
            \State  BS $i$ updates $\gamma^{(k+1)}_i$ according to \eqref{update}, $\forall i \in \mathcal{N}$
		\EndFor
	\end{algorithmic}
\end{algorithm}

\section{Performance Analysis}
In this section, we analyze the performance of GLOBE. To facilitate our exposition, we introduce the following auxiliary notations: $c^{max}_{tx} \triangleq \max_u c_{tx, u}$, $c^{max}_{com} \triangleq \max_u c_{com, u}$, $p^{min} \triangleq \min_{i,u,t} p^t_{i,u}$ and $f^{min} \triangleq \min_i f_i$. Let $c^{max} \triangleq \max\{c^{max}_{tx}/p^{min},c^{max}_{com}/(\kappa (f^{min})^2)\}\}$. Moreover, we assume that $E_i(\bm\alpha^t)$ is upper bounded by $E^{max}_{tx}$ and $E_i(\bm\beta^t)$ is upper bounded by $E^{max}_{com}$. This assumption holds when the communication traffic/computation task arrivals are bounded. Let $E^{max} = E^{max}_{tx} + E^{max}_{com}$.

\begin{lemma} \label{lemma:perturbation}
For any $B^{max} > E^{max} + \mathcal{E}^{max} + g^{max}$, by choosing
\begin{align}
0\leq V \leq \frac{B^{max} - E^{max} - \mathcal{E}^{max} - g^{max}}{c^{max}}
\end{align}
and $\theta = V c^{max} + E^{max}$, the battery level $B^t_i$ is bounded in $[0, B^{max}], \forall i$, and the energy causality constraint is satisfied in every time slot.
\end{lemma}
\begin{proof}
Define $\tilde{\theta}^t_i \triangleq \theta - V\cdot\max_u\max\{c_{tx,u}/p^t_{i,u}, c_{com,u}/(\kappa(f_i)^2)\}$. Using the definition of $c^{max}$, we must have $\tilde{\theta}^t_i  \geq E^{max}$. Consider the following three cases based on the value of $B^t_i$.

Case 1: $B^t_i \in [\theta, B_{max}]$. Our analysis for the optimal energy harvesting and purchase problem before shows that if $B^t_i > \theta$, then BS $i$ will not harvest or purchase any energy, namely $e^{t*}_i = 0$ and $g^{t*}_i = 0$. According to the battery dynamics, $B^{t+1}_i \leq B^t_i$ and hence $B^{t+1}_i$ will also be less than $B_{max}$. On the other hand, because $\theta > E_{max}$, the energy causality constraint $E^{t*}_i < B^t_i$ is satisfied and the next slot battery will be greater than 0 since $B^{t+1}_i = B^t_i - E^{t*}_i > \theta - E^{max} > 0$.

Case 2: $B^t_i \in [\tilde{\theta}^t_i, \theta]$. In this case, we have $e^{t*}_i \leq \mathcal{E}^{max}$, $g^{t*}_i \leq g^{max}$, $E^{t*}_{tx, i} \leq E^{max}_{tx}$ and $E^{t*}_{com, i} \leq E^{max}_{com}$. Therefore,
\begin{align}
B^{t+1}_i \leq \theta + \mathcal{E}^{max} + g^{max} = Vc^{max} + E^{max} + \mathcal{E}^{max} + g^{max} \leq B^{max}
\end{align}
where the second equality is obtained by plugging in the definition of $\theta$, and the last inequality follows the chosen value range of $V$. On the other hand, because $\tilde{\theta}^t_i \geq E^{max}$, the energy causality constraint is satisfied and we have $B^{t+1}_i \geq 0$.

Case 3: $B^t_i \in [0, \tilde{\theta}^t_i]$. We first investigate $E^{t*}_{tx,i}$ and $E^{t*}_{com,i}$. The objective function in the transmission load balancing problem \eqref{TLB} can be rearranged as
\begin{align}
\sum_{i\in\mathcal{N}} \left( \sum_{u\in\mathcal{U}_i}Vc_{tx,u}\sum_{j\in\mathcal{N}_u}\alpha^t_{u,j} + \tilde{B}^t_{i}\sum_{u} p^t_{i,u}\alpha^t_{u,i}\right)
=\sum_{i\in\mathcal{N}}\sum_{u\in \mathcal{U}}(Vc_{tx,u} + \tilde{B}^t_i p^t_{i,u})a^t_{u,i}.
\end{align}
If $Vc_{tx,u} + \tilde{B}^t_i p^t_{i,u} < 0, \forall u$, then the optimal transmission load balancing strategy does not require BS $i$ to transmit any traffic for any user, namely $a^t_{u,i} = 0, \forall u$. This condition is equivalent to $B^t_i < \theta - Vc_{tx,u}/p^t_{i,u}$ using the definition of $\tilde{B}^t_i$. Further, using the definition of $\tilde{\theta}^t_i$, we have $B^t_{i} < \tilde{\theta}^t_i < \theta - Vc_{tx,u}/p^t_{i,u}$. Therefore, BS $i$ does not incur any transmission energy consumption, namely $E^{t*}_{tx,i} = 0$.

Similarly, the objective function in the computation load balancing problem \eqref{CLB} can be rearranged as
\begin{align}
\sum_{i\in\mathcal{N}}\left(\sum_{u\in\mathcal{U}_i}Vc_{com,u}\sum_{j\in\mathcal{N}_u}\beta^t_{u,j} + \tilde{B}^t_i \kappa (f_i)^2\sum_u \alpha^t_{u,i} \right)
= \sum_{i\in\mathcal{N}}\sum_{u\in\mathcal{U}}(Vc_{com,u} +\tilde{B}^t_i\kappa(f_i)^2)\beta^t_{u,i}.
\end{align}
If $Vc_{com,u} +\tilde{B}^t_i\kappa(f_i)^2 < 0, \forall u$, then the optimal computation load balancing strategy does not require BS $i$ to perform any computation for any user, namely $\beta^t_{u,i} = 0, \forall u$. This condition is equivalent to $B^t_i < \theta - Vc_{com,u}/k(f_i)^2$. Using the definition of $\tilde{\theta}^t_i$, we  have this condition satisfied. Therefore, BS $i$ does not incur any computation energy consumption, namely $E^{t*}_{com,i} = 0$.

Because we have $e^{t*}_i \leq \mathcal{E}^{max}$ and $g^{t*}_i \leq g^{max}$, clearly the energy causality constraint is satisfied, $B^{t+1}_i \geq 0$ and $B^{t+1}_i \leq B^{max}$.
\end{proof}

Lemma 2 is of significant importance because it shows that GLOBE not only yields a feasible solution to the relaxed problem \textbf{P2} but also a feasible solution to the original problem \textbf{P1} since the energy causality constraint in each time slot is actually satisfied by running GLOBE, provided that the battery capacity is sufficiently large and the algorithm parameters are properly chosen. This result helps us to prove the performance guarantee of the proposed GLOBE algorithm.

Next we proceed to show the asymptotic optimality of the GLOBE algorithm, for which we first define the Lyapunov function as follows:
\begin{align}
\Psi^t \triangleq\frac{1}{2}\sum_{i\in\mathcal{N}}(\tilde{B}^t_i)^2=\frac{1}{2}\sum_{i\in\mathcal{N}}\left(B^t_i - \theta\right)^2.
\end{align}
The Lyapunov drift represents the expected change in the Lyapunov function from one time slot to another, which is defined as $\Delta^t = \mathbb{E}[\Psi^{t+1} - \Psi^t|\B^t]$, where the expectation is taken with respect to the random process associated the system, given the battery state $\B^t=[B^t_1,...,B^t_N]$. Assuming for now that the battery capacity is infinite, the battery state dynamics yields
\begin{align}
B^{t+1}_i - \theta = B^t_i - \theta - E^t_i(\bm\alpha^t, \bm\beta^t) + e^t_i + g^t_i, \forall i.
\end{align}
Squaring both sides of the above equation, we obtain,
\begin{align}
(B^{t+1}_i - \theta)^2 = (B^t_i - \theta)^2 + (E_i^t(\bm\alpha^t, \bm{\beta}^t) - e^t_i - g^t_i)^2  - 2(B^t_i - \theta)(E_i^t(\bm\alpha^t, \bm{\beta}^t) - e^t_i - g^t_i).
\end{align}
Notice that the term $(E_i^t(\bm\alpha^t, \bm{\beta}^t) - e^t_i - g^t_i)^2 \leq (E_i^t(\bm\alpha^t, \bm{\beta}^t))^2 + (e^t_i + g^t_i)^2 \leq (E_{tx}^{max} + E_{com}^{max})^2 + (\mathcal{E}^{max} + g^{max})^2 \triangleq 2D/N$ is upper-bounded by a constant $2D/N$. Using this bound and rearranging the above equation, we have
\begin{align}
(B^{t+1}_i - \theta)^2 - (B^t_i - \theta)^2
\leq 2D/N - 2(B^t_i - \theta)(E_i^t(\bm\alpha^t, \bm{\beta}^t) - e^t_i - g^t_i).
\end{align}
Using the above inequality and the definition of $\Delta^t$, we have
\begin{align}
\Delta^t \leq D - \mathbb{E}[\sum_{i\in\mathcal{N}}(B^t_i - \theta)(E_i^t(\bm\alpha^t, \bm{\beta}^t) - e^t_i - g^t_i)|\B^t].
\end{align}
Adding the system cost multiplied by $V$, namely $V\cdot\mathbb{E}[\sum_{i\in\mathcal{N}} C_i(\bm\alpha^t, \bm\beta^t, g^t_i)|\B^t]$, to both sides and denoting $\Delta_V^t = \Delta^t + V\mathbb{E}[\sum_{i\in\mathcal{N}} C_i(\bm\alpha^t, \bm\beta^t, g^t)|\B^t]$, we have
\begin{align}
\Delta_V^t \leq D + \mathbb{E}[\sum_{i\in\mathcal{N}}\left(V C_i(\bm\alpha^t, \bm\beta^t, g^t_i) - \tilde{B}^t_i (E_i^t(\bm\alpha^t, \bm{\beta}^t) - e^t_i -g^t_i) \right)|\B^t].  \label{bound}
\end{align}

According to the theory of Lyapunov optimization (drift-plus-penalty method), the control actions are chosen for each time slot $t$ to minimize the bound on the modified Lyapunov drift function $\Delta^t_V$. Therefore, in each time slot $t$, we solve the per-time slot optimization problem \textbf{P3} to obtain load balancing strategies and energy harvesting and purchase strategies as in GLOBE.

Theorem 1 provides the theoretical performance guarantee of GLOBE.

\begin{theorem}\label{onlinealg_bound}
For any $V$, if the battery capacity satisfies $B^{max} \geq V c^{max} + E^{max} + \mathcal{E}^{max} + g^{max}$, then the proposed algorithm yields a feasible solution and the achievable time average system cost satisfies
\begin{align}
\lim_{T\to\infty} \frac{1}{T}\sum_{t=1}^T \sum_{i\in\mathcal{N}} \mathbb{E}\left[C_i(\bm\alpha^t, \bm\beta^t, g^t_i)\right] \leq C^*_1 + D/V
\end{align}
where $D$ is a constant.
\end{theorem}
\begin{proof}
Consider the bound on the Lyapunov drift function \eqref{bound}. It is clear that the control actions $\bm\alpha^t, \bm\beta^t, \e^t, \g^t$ by our algorithm minimizes the bound on the Lyapunov function over all possible control actions. Comparing it with the control actions chosen according to the optimal oracle policy that achieves $C^*_2$, we have
\begin{align}
&\Delta^t + V\mathbb{E}[\sum_{i\in\mathcal{N}} C_i(\bm\alpha^t, \bm\beta^t, g^t_i)|\B^t] \nonumber\\
\leq& D + \mathbb{E}[\sum_{i\in\mathcal{N}}\left(V C_i(\bm\alpha^t, \bm\beta^t, g^t_i) - \tilde{B}^t_i (E_i^t(\bm\alpha^t, \bm{\beta}^t) - e^t_i -g^t_i) \right)|\B^t] \nonumber\\
\leq& D + \mathbb{E}[\sum_{i\in\mathcal{N}}\left(V C_i(\bm\alpha^{\Pi,t}, \bm\beta^{\Pi,t}, g^t_i) - \tilde{B}^t_i (E_i^t(\bm\alpha^{\Pi,t}, \bm{\beta}^{\Pi,t}) - e^{\Pi,t}_i -  g^{\Pi, t}_i)\right)|\B^t] \nonumber\\
=&V\mathbb{E}[\sum_{i\in\mathcal{N}} C_i(\bm\alpha^{\Pi,t}, \bm\beta^{\Pi,t}, g^{\Pi, t}_i)] + D - \mathbb{E}[\tilde{B}^t_i (E_i^t(\bm\alpha^{\Pi,t}, \bm{\beta}^{\Pi,t}) - e^{\Pi,t}_i - g^{\Pi_t}_i)|\B^t]. \nonumber
\end{align}
Plugging in \eqref{lemmaC} and \eqref{lemmaE}, taking the expectation on both sides and summing from $t = 0, ..., T-1$, normalizing by $T$ and taking the limit $T\to\infty$, we have
\begin{align}
V\lim_{T\to\infty} \frac{1}{T}\sum_{t=1}^T\sum_{i\in\mathcal{N}}\mathbb{E} \left[C_i(\bm\alpha^t, \bm\beta^t)\right] \leq V C^*_2 + D,
\end{align}
where $D < \infty$ is a constant. The proof is completed with $C^*_2\leq C^*_1$.
\end{proof}
Theorem 1 proves that GLOBE can achieve the minimum cost achievable by the genie-aided offline algorithm within a bounded deviation, without foreseeing the future information. Moreover, it formalizes a critical tradeoff between the battery capacity and the achievable system performance: the achievable system performance improves with the increase of the battery capacity. In particular, the system performance can be made arbitrarily close to optimum if the battery capacity is large enough. This result provides profound guidelines for EH-powered MEC network design and deployment, especially on the  battery design.

\section{Simulation Results}
In this section, we evaluate the performance of GLOBE through simulations. We consider $N=5$ BSs, who are able to perform computation offloading to only some other BSs. To be specific, we let $\mathcal{M}_i\in \mathcal{N}$ and $|\mathcal{M}_i|=3, \forall i$. The downlink data traffic arrival at BS $i$  is modeled as a Poisson process with $\mu_i^t\in[0,10]$ unit/sec. The expected size of each data traffic $\mathbb{E}[\omega]$ is set as 100 Mbits. The transmitting power of BS $i$ is $P_{tx,i}=1$ W, the noise power spectral density is $\sigma^2=0.01$ W/Hz, and the bandwidth is $W=20$ MHz. The downlink traffic dropping costs are $c_{tx,i}=10, \forall i$. The computation task arrival at BS $i$ is modeled as a Poisson process with $\lambda^t_i\in[0, 10]$ with the mean task size of 1 M unit. The CPU speed is $f_i= 2.4$ GHz and the expected number of CPU cycles required for each computation task is $\rho=8\times10^5$. The computation constraint is $d^{max}=1$ ms. The energy consumption parameter is chosen as $\kappa=2.5\times10^{-22}$. The cost of dropping one unit of computation task is $c_{com,i}=0.01$. The harvested energy is modeled as a uniform distribution and satisfies $\mathcal{E} \in [0,10]$. The unit energy price of the alternative grid power $c_{grid}^t$ is modeled as a uniform distribution with mean 1 if not specified. The proposed GLOBE is compared with following three benchmark solutions:
\begin{itemize}
	\item Non-GLB Stochastic Optimization (\textbf{SO-NG}): SO-NG considers the stochasticity in communication traffic, computation task arrivals and energy harvesting (i.e., battery dynamics and energy causality) and aims to minimize the long-term system cost. However, in this case, GLB is not enabled in the network. This problem can be solved online by leveraging the \emph{Lyapunov Optimization with Perturbation} similar to that in GLOBE.
	
	\item Myopic Optimization with GLB (\textbf{MO-G}): MO-G optimizes GLB and admission control to minimize the system cost in the current time slot. The decisions are made without concerns of the future energy harvesting and forthcoming communication/computation workload.
	
	\item Non-GLB Myopic Optimization (\textbf{MO-NG}): MO-NG is the most naive scheme which does not consider the geographical load balancing or the long-term system performance. Each BS simply tries to serve all computation workload and communication traffic given the available energy in the battery and drops whatever cannot be fulfilled.
	
\end{itemize}


\subsection{Runtime Performance Comparison}

Fig. \ref{fig:runtime_comp} compares the runtime performance of GLOBE and three benchmarks. We focus on two metrics: the \emph{time-average cost} of BSs in Fig. \ref{fig:Tave_cost} and the \emph{time-average battery level} in Fig. \ref{fig:Tave_battery}. It can be observed from Fig. \ref{fig:Tave_cost} that generally GLOBE achieves the lowest long-term system cost compared to the other three benchmarks. Specifically, GLOBE reduces the system cost by nearly 50\% compared to MO-NG and nearly 30\% compared to the second best scheme, i.e., SO-NG. Moreover, we see that addressing both issues (i.e., loading balancing of communication/computation and stochastic optimization for long-term system performance) provides considerable improvements to the system performance. Comparing the time-average costs of GLOBE and SO-NG (or MO-G and MO-NG), it can be concluded that enabling load balancing of communication traffic and computation workload among BSs helps to reduce the system cost. Comparing the time-average costs of GLOBE and MO-G, we see that a significant system cost reduction can be achieved in the long-run by carefully scheduling the energy consumption, harvesting and purchasing in each time slot. Another important feature of GLOBE is that it ensures the performance with bounded battery levels, thereby enabling practical implement. As can be seen in Fig. \ref{fig:Tave_battery}, the time-average battery levels of GLOBE and \emph{SO-NG} stabilize at a relatively low value by following the online decisions designed by \emph{Lyapunov optimization with perturbation}. This means that the performance of GLOBE is achievable with a small battery capacity. By contrast, the two myopic benchmarks require a very large battery capacity to implement these schemes.

\begin{figure}[h]
	\centering	
	\subfigure[System cost of BSs]{\label{fig:Tave_cost}
		\includegraphics[width= 0.4 \linewidth]{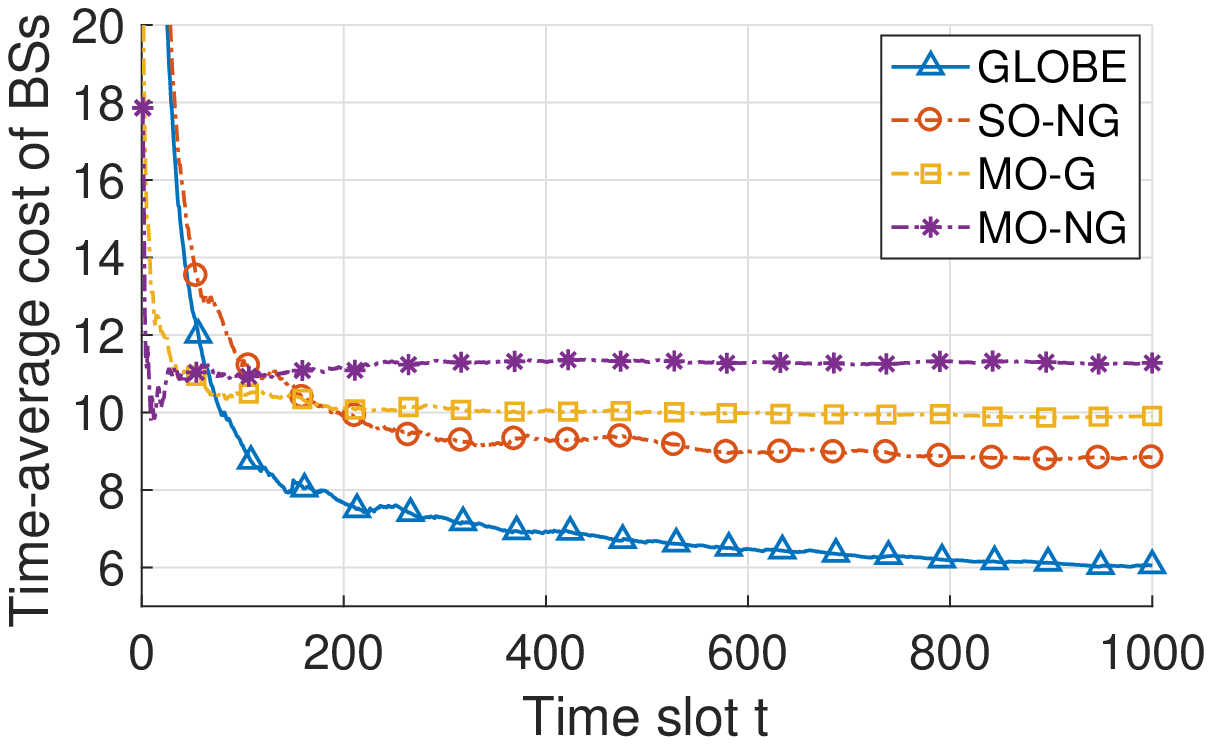}}
	\subfigure[Battery level of BSs]{\label{fig:Tave_battery}
		\includegraphics[width=0.4 \linewidth]{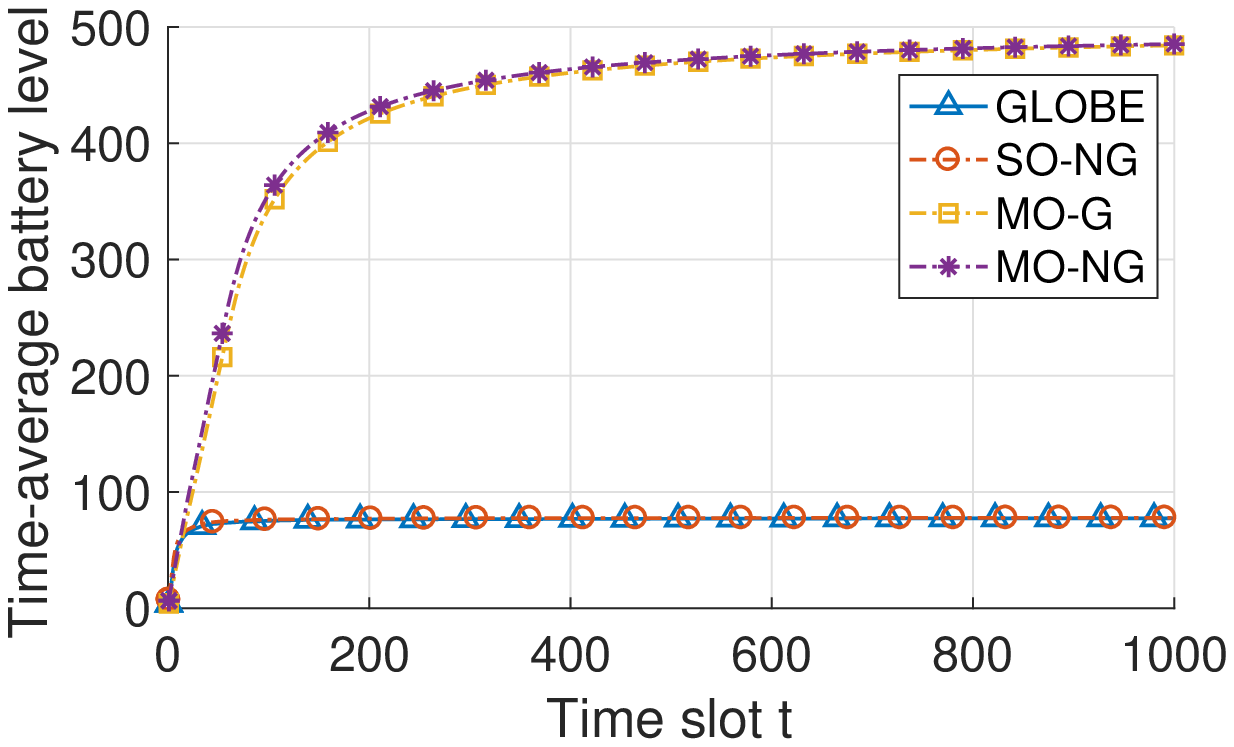}}
\vspace{-10pt}
	\caption{Runtime performance comparison of various solutions.}
	\label{fig:runtime_comp}
\vspace{-30pt}
\end{figure}

\subsection{Impact of Control Parameter $V$}
Fig. \ref{fig:impact_V} depicts the time-average system cost and time-average battery level achieved by GLOBE with various values of $V$. It shows that the system cost decreases with the increase in $V$. This is because a larger $V$ empathizes the cost minimization more in problem \textbf{P3}. However, a lower system cost is achieved at the price of a higher requirement on battery capacity. As shown in Fig. \ref{fig:Tave_cost_V}, the time-average battery level stabilizes at a higher value with a larger $V$, which means that a larger battery capacity is required to implement the algorithm.

Fig. \ref{fig:tradeoff} formally presents the trade-off between the system cost and the battery capacity. It clearly shows a $[O(1/V),O(V)]$ trade-off between the stabilized time-average cost and battery capacity, which is consistent with our analysis in Theorem \ref{onlinealg_bound}.  Moreover, we see that the stabilized battery level of BSs closely follows the designed perturbation parameter $\theta$.

\begin{figure}[h]
	\centering	
	\subfigure[Impact on system cost]{\label{fig:Tave_cost_V}
		\includegraphics[width= 0.4 \linewidth]{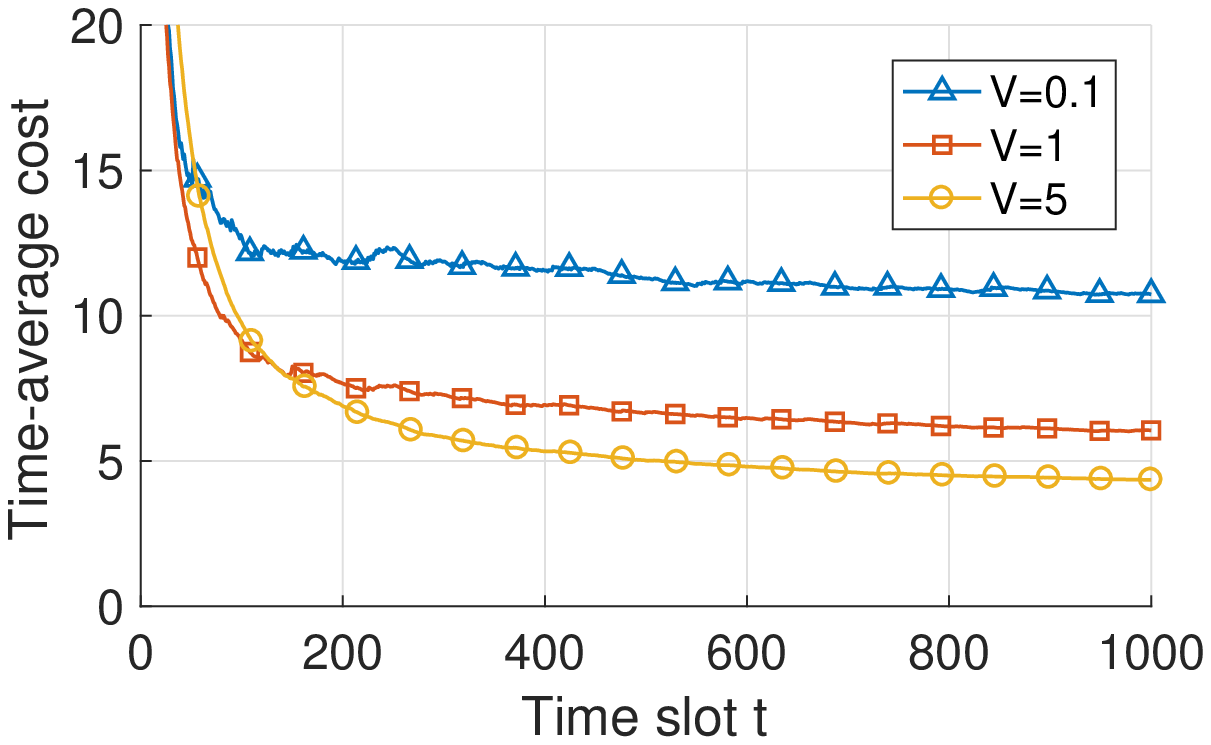}}
	\subfigure[Impact on battery level]{\label{fig:Tave_battery_V}
		\includegraphics[width= 0.4 \linewidth]{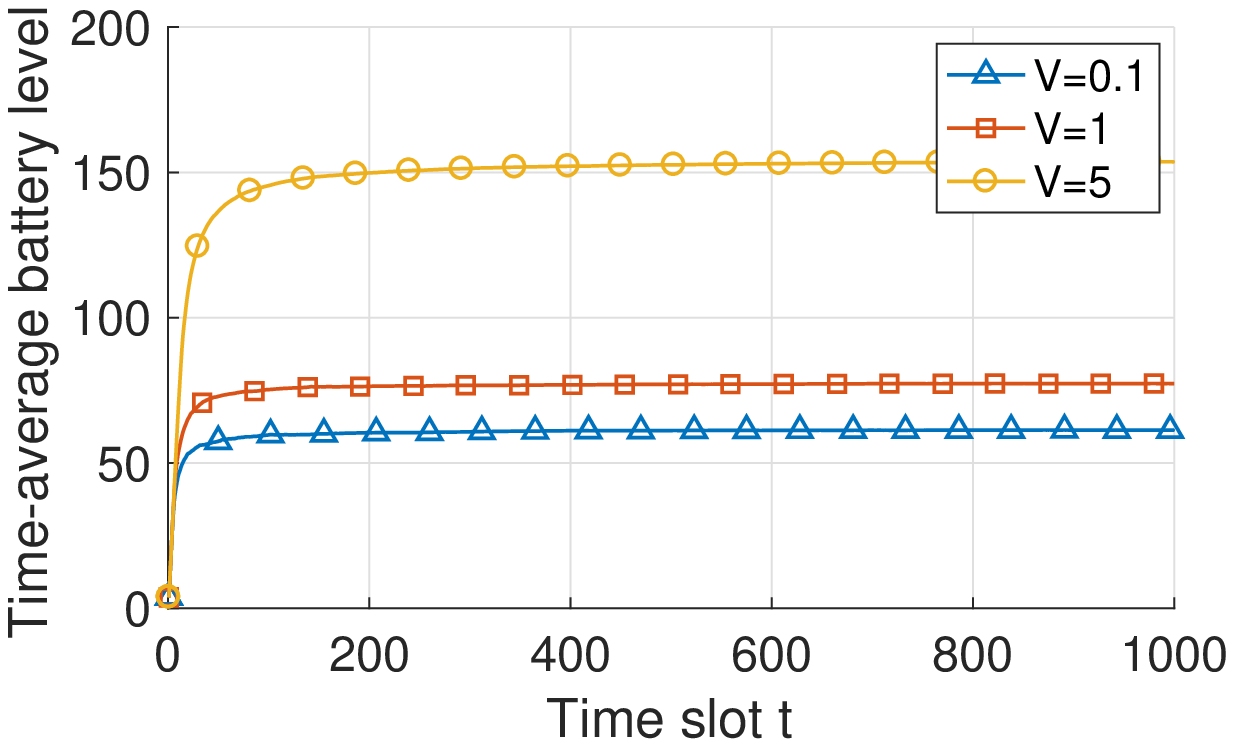}}
	\caption{Impact of control parameter $V$ to the performance of GLOBE.}
	\label{fig:impact_V}
\vspace{-20pt}
\end{figure}

\begin{figure}[h]
	\centering	
	\includegraphics[width=0.45 \linewidth]{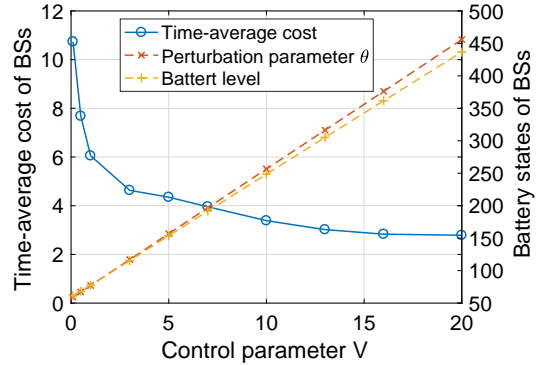}
	\vspace{-10pt}
	\caption{Trade-off between system cost and battery capacity.}
	\label{fig:tradeoff}
\vspace{-20pt}
\end{figure}

\subsection{Impact of Grid Power}
We also evaluate the impact of grid power on the algorithm performance. Fig. \ref{fig:grid_power} reports the performance of GLOBE with or without grid power. Moreover, we vary the expected market energy price to evaluate the role of energy price in GLOBE. We see from Fig. \ref{fig:Tave_cost_grid} that the system cost is reduced  with the supplement of grid power. This is due to the fact that the cost of purchasing grid power is usually lower than that of dropping communication/computation workloads, and hence the grid power can be used as a supplementary energy source when the energy harvesting is not sufficient. In the worst case, if the market energy purchasing costs more than the workload dropping, then GLOBE will simply choose not to purchase. Fig. \ref{fig:Tave_cost_grid} further indicates that the system cost is reduced along with the decrease in the market energy price since less fee is charged for the required grid power.

Fig. \ref{fig:Tave_battery_grid} depicts the time-average battery levels of these four cases. It is observed that the time-average battery levels of cases with grid power stabilize around the same value, i.e., the perturbation parameter. This is due to the fact that the design of perturbation parameter is independent of the marker energy price (as shown in Lemma \ref{lemma:perturbation}). For the case with no grid power, the stabilized battery level is slightly lower. This is because the BSs cannot maintain the desired battery level for future use without the grid power, since the energy harvesting may be insufficient at times.

\begin{figure}[h]
	\centering	
	\subfigure[Impact on system cost]{\label{fig:Tave_cost_grid}
		\includegraphics[width=0.45 \linewidth]{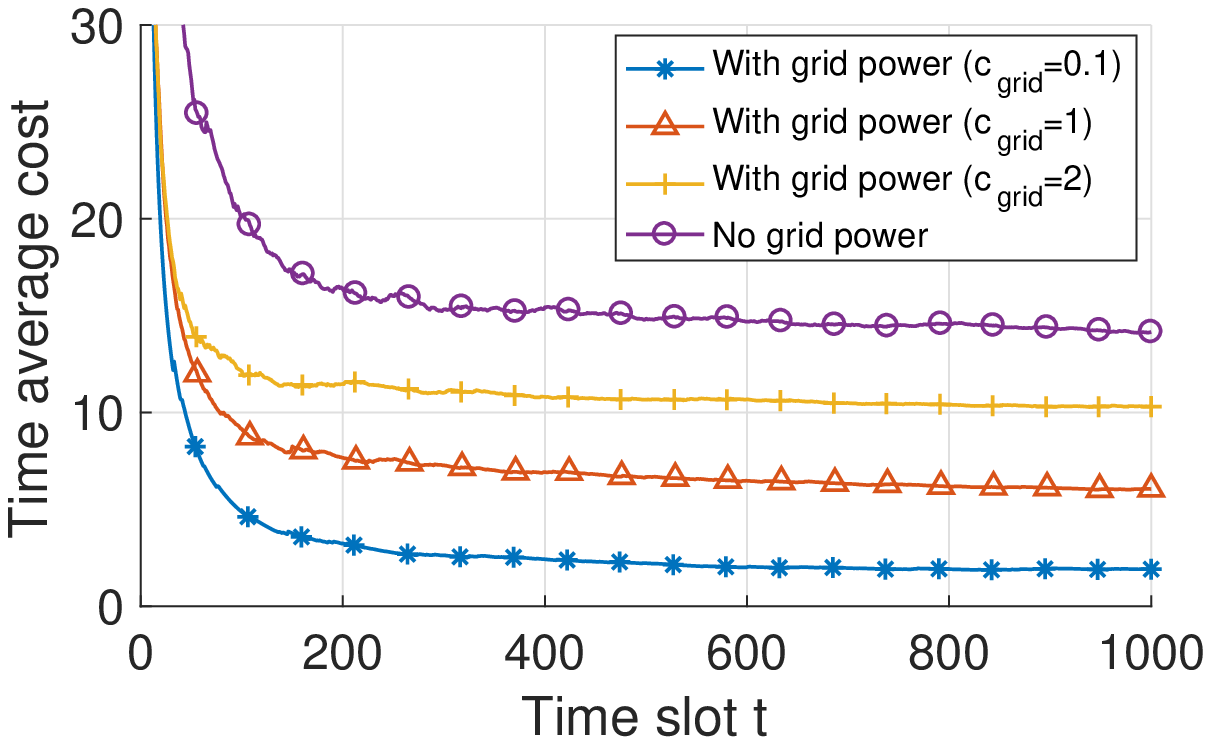}}
	\subfigure[Impact on battery level]{\label{fig:Tave_battery_grid}
		\includegraphics[width=0.45 \linewidth]{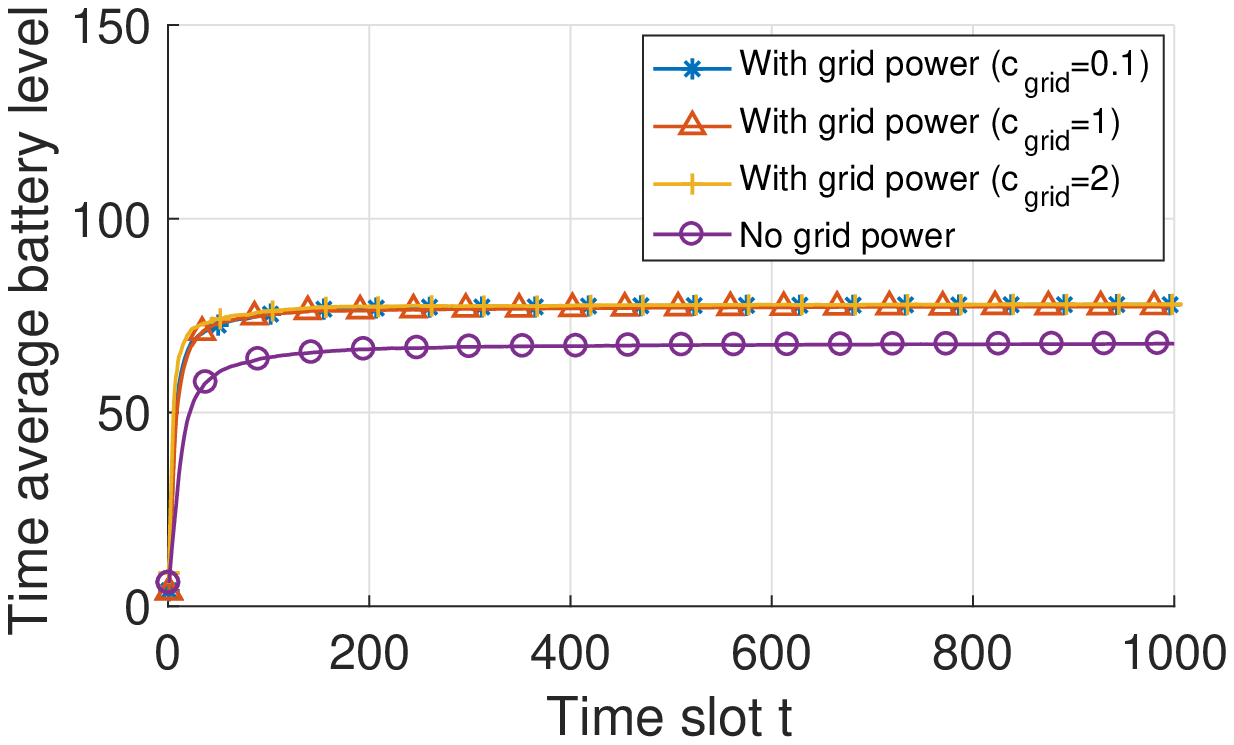}}
	\caption{Impact of grid power to the performance of GLOBE.}
	\label{fig:grid_power}
\vspace{-20pt}
\end{figure}

\subsection{Impact of Computation Workload Intensity}
Fig. \ref{fig:workload_intensity} depicts the time-average system costs with different levels of computation workload intensity. In general, the system cost grows with the increase in computation workload since larger computation workload incurs higher delay cost and energy consumption. We also see that if the computation workload intensity becomes too low or too high, the system cost achieved by GLOBE is almost the same as that achieved by SO-NG. This is reasonable because if the computation workload is extremely low at every BS, then the BSs can process the workload \emph{locally} with their own resources, while incurring  low delay cost and energy consumption. Hence, there is no need for load balancing. On the other hand, if the computation workload intensity is too high, every BS has already been overloaded, therefore, there is little benefit from loading balancing among BSs. A large system cost reduction is achieved when computation workload intensity is at the similar level of the system computation capacity. In this case, there are both many overloaded BSs and underloaded BSs in the system, and hence the geographical load balancing can better help to reduce the system cost.
\begin{figure}[h]
	\centering	
	\includegraphics[width=0.45 \linewidth]{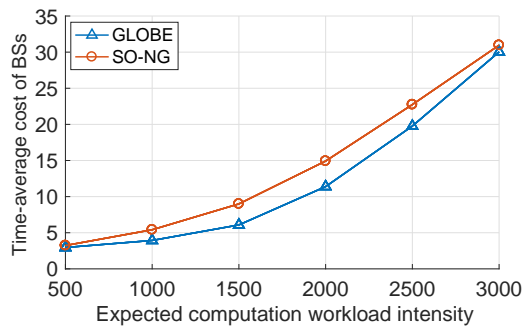}
\vspace{-10pt}
	\caption{Impact of computation workload intensity to the performance of GLOBE and SO-NG.}
	\label{fig:workload_intensity}
\vspace{-20pt}
\end{figure}

\subsection{Before versus After GLB}
Fig. \ref{fig:GLB} shows the distribution of computation workload and communication traffic before and after applying GLB in one particular time slot. It is worth noticing that the GLOBE does not simply balance the workload evenly among BSs. Instead, GLOBE distributes the workload based on the battery states of BSs. As can be seen in  Fig. \ref{fig:GLB}, BSs with higher battery level (e.g. BS 1 and BS 5) tend to serve more workload and traffic, which helps BSs with low battery level to avoid  battery depletion by conservatively serving less workload, thereby reducing the probability of purchasing energy from the power grid and decreasing the system cost.
\begin{figure}[h]
	\centering	
	\subfigure[GLB for computation workload]{\label{fig:GLB_computation}
		\includegraphics[width= 0.4 \linewidth]{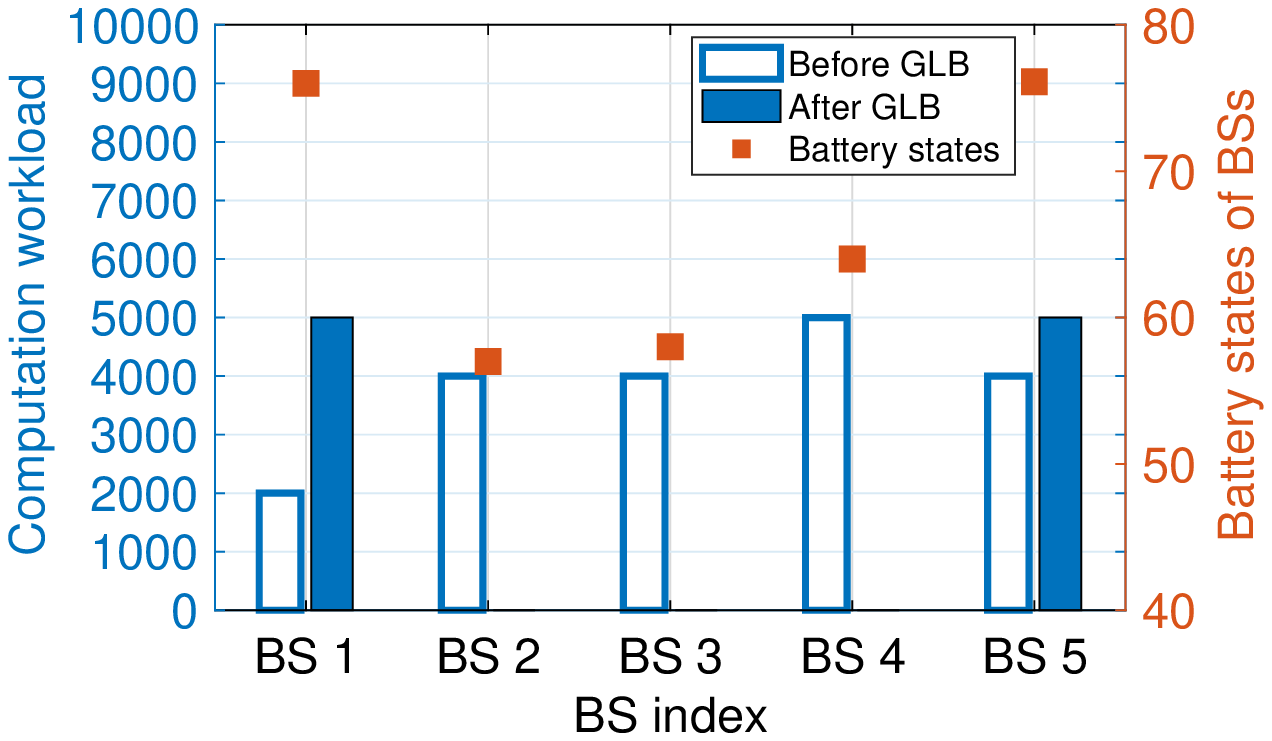}}
	\subfigure[GLB for communication traffic]{\label{fig:GLB_communication}
		\includegraphics[width= 0.4 \linewidth]{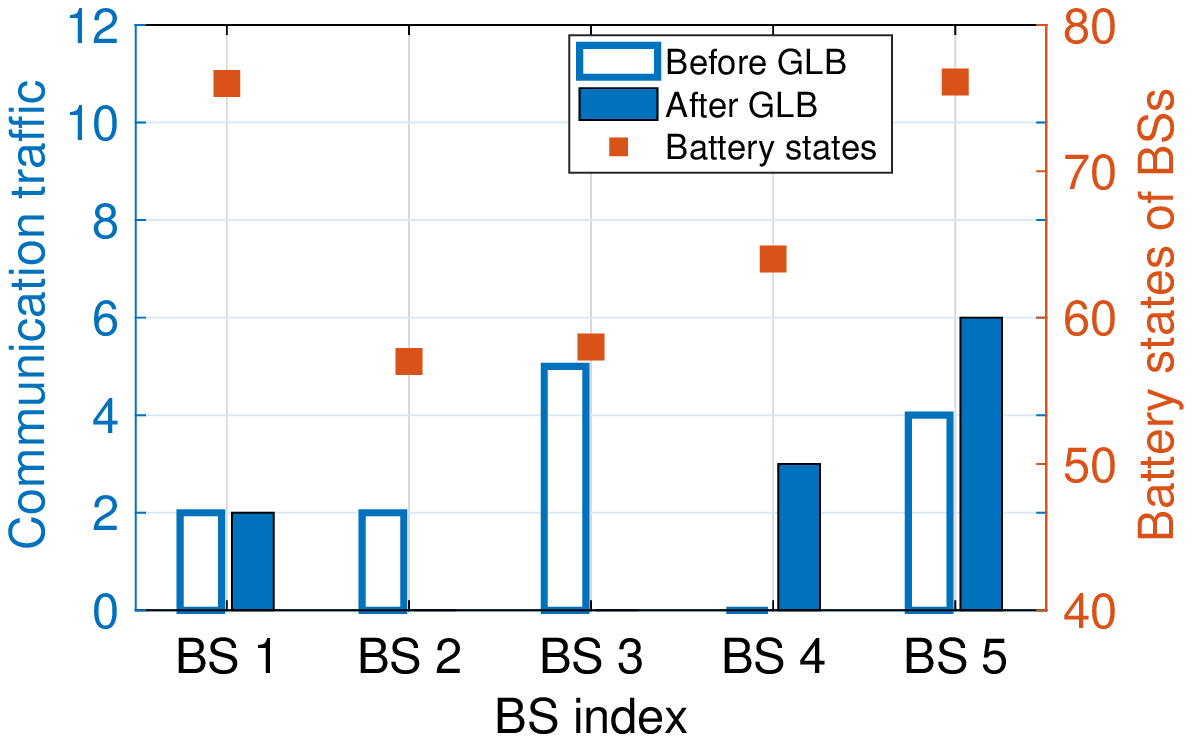}}
	\caption{Workload distribution before and after GLB.}
	\label{fig:GLB}
\vspace{-20pt}
\end{figure}

\subsection{Convergence of Distributed Algorithm}
Fig. \ref{fig:convergence} depicts the achieved objective values by substituting the solutions derived by the distributed algorithm into the objective functions of primal LP in \eqref{CLB} and regularized QP in \eqref{CLB regularized}. There are several observations worth pointing out. First, It can be seen from Fig. \ref{fig:QP_convergence} that the proposed distributed algorithm ensures the convergence to the optimal solution of the regularized QP and the optimal objective value of regularized QP is 329.60 in this particular case. Second, Fig. \ref{fig:LP_convergence} shows that the optimal solution to the regularized QP is exactly the optimal solution to the primal LP since this solution also achieves the optimal objective value (i.e., 331.06) of the primal LP problem.

\begin{figure}[h]
	\centering	
	\subfigure[Convergence of QP during iteration]{\label{fig:QP_convergence}
		\includegraphics[width=2.5 in]{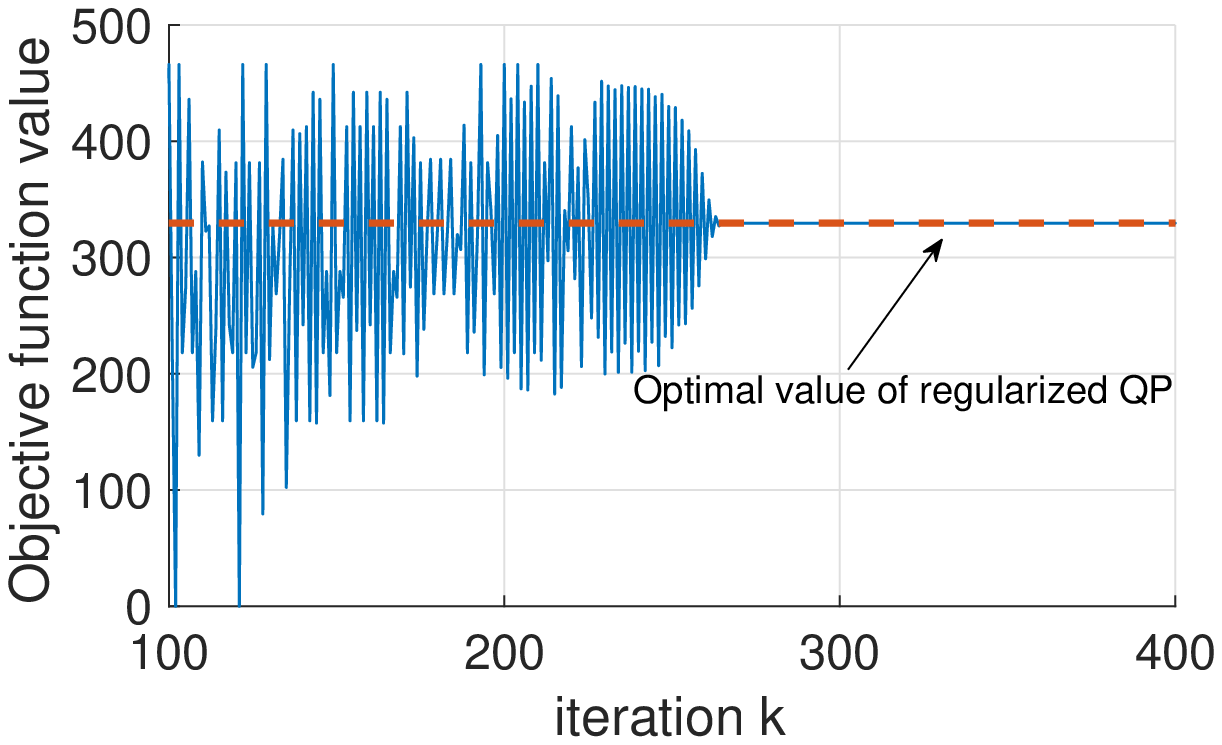}}
	\subfigure[Convergence of LP during iteration]{\label{fig:LP_convergence}
		\includegraphics[width=2.5 in]{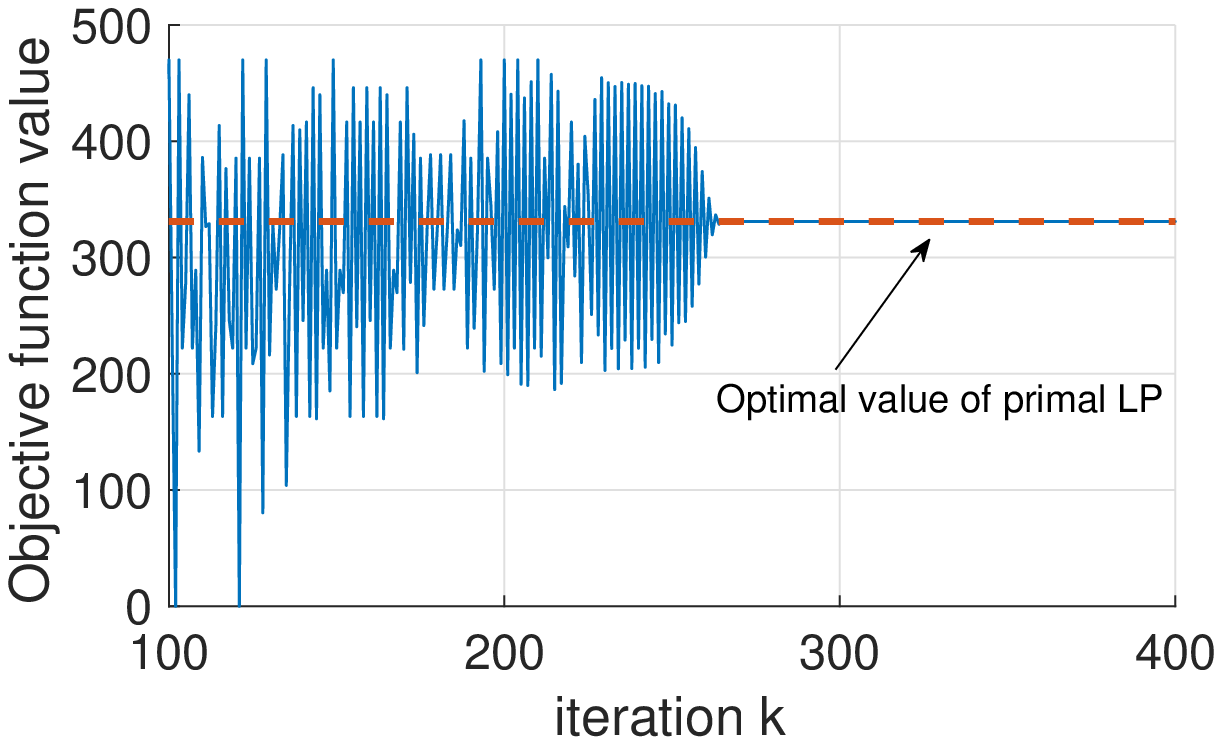}}
	
	\caption{Convergence and optimality of the distributed algorithm.}
	\label{fig:convergence}
\vspace{-20pt}
\end{figure}

Note that the distributed algorithm may not always obtain the optimal solution to the original LP problem. However, even when it converges to a suboptimal solution, the gap between the converged suboptimal value and the optimal value is very small. Among all 1,000 time slots in our simulation, we only observed 3 time slots that have an error percentage larger than 0.5\% and the maximum error percentage is 3\%, which can be reasonably neglected.

Fig. \ref{fig:cen_dis_comp} compares the performances of distributed and the centralized algorithms. We see that the performance achieved by the distributed algorithm is almost identical to that achieved by the centralized algorithm. This further validates that the proposed GLOBE algorithm can be effectively implemented in a distributed manner at each BS without sacrificing the performance.

\begin{figure}[h]
	\centering	
	\includegraphics[width=0.45 \linewidth]{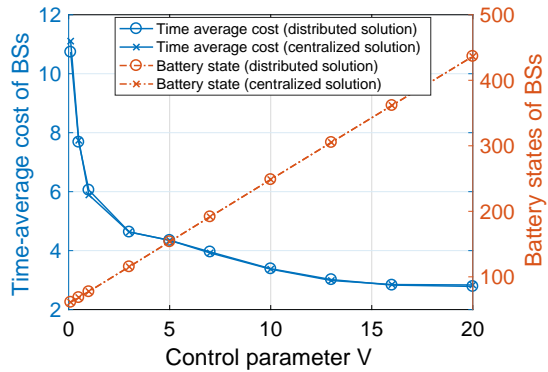}
\vspace{-10pt}
	\caption{Performance comparison between distributed and centralized algorithms.}
	\label{fig:cen_dis_comp}
\vspace{-20pt}
\end{figure}

\section{Conclusion}
In this paper, we have proposed an online algorithm to perform geographical load balancing in EH-powered MEC networks. We demonstrated that a fundamentally new design that simultaneously manages the limited energy, computing and radio access resources in both spatial and temporal domains is key to fully reaping the benefits of EH-power MEC. The proposed GLOBE algorithm operates online without the need to acquiring future system information. Moreover, GLOBE can be implemented in a distributed manner, where BSs solve local optimization problems with very limited information exchange. Our algorithm is simple and easy to implement in practical deployment scenarios, yet provides provable performance guarantee. Comprehensive numerical simulations have been carried out to validate the theoretical analysis and illustrate the performance improvement over other benchmark solutions.

\bibliographystyle{IEEEtran}
\bibliography{refs}

\end{document}